\def\confversion{0}
\def\ifconf{\ifnum\confversion=1}
\def\ifnotconf{\ifnum\confversion=0}

\documentclass[11 pt]{article}

\def\showauthornotes{0}
\def\showkeys{0}
\def\showdraftbox{1}

\usepackage{xspace,xcolor,enumerate}
\usepackage{amsmath,amssymb}
\usepackage{amsthm}
\usepackage[toc,page]{appendix}
\usepackage{thmtools}
\usepackage{thm-restate}
\usepackage{color,graphicx}
\usepackage{boxedminipage}
\usepackage{makecell}
\usepackage{tabularx}
\usepackage{enumitem}
\usepackage[linesnumbered,ruled,vlined]{algorithm2e}
\ifnum\showkeys=1
\usepackage[color]{showkeys}
\fi

\definecolor{darkred}{rgb}{0.5,0,0}
\definecolor{darkgreen}{rgb}{0,0.35,0}
\definecolor{darkblue}{rgb}{0,0,0.55}

\usepackage[pdfstartview=FitH,pdfpagemode=UseNone,colorlinks,linkcolor=darkblue,filecolor=darkred,citecolor=darkgreen,urlcolor=darkred,pagebackref]{hyperref}

\usepackage[capitalise,nameinlink]{cleveref}
\usepackage[T1]{fontenc}
\usepackage{mathtools,dsfont,bbm}
\usepackage{mathpazo}

\ifnum\showauthornotes=1
\newcommand{\Authornote}[2]{{\sf\small\color{red}{[#1: #2]}}}
\newcommand{\Authorcomment}[2]{{\sf \small\color{gray}{[#1: #2]}}}
\newcommand{\Authorfnote}[2]{\footnote{\color{red}{#1: #2}}}
\else
\newcommand{\Authornote}[2]{}
\newcommand{\Authorcomment}[2]{}
\newcommand{\Authorfnote}[2]{}
\fi

\ifnum\showdraftbox=1

\else

\fi

\newtheorem{theorem}{Theorem}[section]

\newtheorem{definition}[theorem]{Definition}

\newtheorem{lemma}[theorem]{Lemma}
\newtheorem{remark}[theorem]{Remark}

\newtheorem{corollary}[theorem]{Corollary}
\newtheorem{claim}[theorem]{Claim}
\newtheorem{fact}[theorem]{Fact}

\theoremstyle{remark}
\newtheorem{algo}[theorem]{Algorithm}

\newenvironment{algorithm_simple}[3]
        {\noindent\begin{boxedminipage}{\textwidth}
        \begin{algo}[#1]\ \par
        {\begin{tabular}{r l}
        \textbf{Input} & #2\\
        \textbf{Output} & #3
        \end{tabular}\par\enskip}}
        {\end{algo}\end{boxedminipage}}

\def\FullBox{\hbox{\vrule width 6pt height 6pt depth 0pt}}

\def\qed{\ifmmode\qquad\FullBox\else{\unskip\nobreak\hfil
\penalty50\hskip1em\null\nobreak\hfil\FullBox
\parfillskip=0pt\finalhyphendemerits=0\endgraf}\fi}

\def\qedsketch{\ifmmode\Box\else{\unskip\nobreak\hfil
\penalty50\hskip1em\null\nobreak\hfil$\Box$
\parfillskip=0pt\finalhyphendemerits=0\endgraf}\fi}

\newcommand{\acts}{{\cdot}}

\def\to{\rightarrow}

\def\epsilon{\varepsilon}

\def\phi{\varphi}

\newcommand{\ie}{i.e.,\xspace}
\newcommand{\eg}{e.g.,\xspace}
\newcommand{\etal}{et al.\xspace}

\newcommand{\R}{{\mathbb R}}
\newcommand{\E}{{\mathbb E}}
\newcommand{\C}{{\mathbb C}}

\newcommand{\Z}{{\mathbb Z}}

\newcommand{\F}{{\mathbb F}}

\newcommand{\cC}{\mathcal{C}}

\newcommand{\cH}{\mathcal{H}}
\newcommand{\cD}{\mathcal{D}}

\usepackage{nicefrac}

\newcommand{\abs}[1]{\ensuremath{\left\lvert #1 \right\rvert}}

\newcommand{\norm}[1]{\ensuremath{\left\lVert #1 \right\rVert}}

\makeatletter

\def\ProbabilityRender#1#2{
  \@ifnextchar\bgroup%
  {\renderwithdist{#1}{#2}}
   {\singlervrender{#1}{#2}}
}
\def\singlervrender#1#2{%
   \ensuremath{\mathchoice
       {{#1}\left[ #2 \right]}
       {{#1}[ #2 ]}
       {{#1}[ #2 ]}
       {{#1}[ #2 ]}
   }
}
\def\renderwithdist#1#2#3{%
   \@ifnextchar\bgroup
   {\superfancyrender{#1}{#2}{#3}}
   {\ensuremath{\mathchoice
      {\underset{#2}{#1}\left[ #3 \right]}
      {{#1}_{#2}[ #3 ]}
      {{#1}_{#2}[ #3 ]}
      {{#1}_{#2}[ #3 ]}
     }
   }
}
\def\superfancyrender#1#2#3#4#5{
   \ensuremath{\mathchoice
      {\underset{#1}{{#1}}\left#4 #3 \right#5}
      {{#1}_{#2}#4 #3 #5}
      {{#1}_{#2}#4 #3 #5}
      {{#1}_{#2}#4 #3 #5}
   }
}
\makeatother

\newfont{\inhead}{eufm10 scaled\magstep1}

\newcommand{\poly}{{\mathrm{poly}}}
\newcommand{\polylog}{{\mathrm{polylog}}}

\DeclareMathOperator\supp{Supp}

\DeclareMathOperator{\tr}{\operatorname {tr}}

\DeclareMathOperator{\aut}{\operatorname{Aut}}

\DeclareMathOperator{\diag}{\operatorname{diag}}

\renewcommand{\Re}{\operatorname{Re}}
\renewcommand{\Im}{\operatorname{Im}}

\newcommand{\inparen}[1]{\left(#1\right)}             

\DeclareSymbolFont{extraup}{U}{zavm}{m}{n}
\DeclareMathSymbol{\varheart}{\mathalpha}{extraup}{86}
\DeclareMathSymbol{\vardiamond}{\mathalpha}{extraup}{87}

\DeclareMathOperator{\Sym}{Sym}

\DeclareMathOperator{\spec}{Spec}

\DeclarePairedDelimiter\set{\lbrace}{\rbrace}

\DeclareMathOperator{\bias}{bias}

\usepackage{tikz}

\usepackage{tikz}
\usetikzlibrary{shapes.symbols}
\usepackage{float}
\usepackage{graphicx}
\usepackage{todonotes}

\usepackage[top=1in, bottom=1in, left=1.25in, right=1.25in]{geometry}

\begin{document}

\title{Explicit Abelian Lifts and Quantum LDPC Codes}

\author{Fernando Granha Jeronimo\thanks{{\tt IAS}. {\tt granha@ias.edu}. This material is based upon work supported by the National Science Foundation under Grant No. CCF-1900460 and also supported in part by NSF grant CCF-1816372.} \and Tushant Mittal\thanks{{\tt UChicago}. {\tt tushant@uchicago.edu}. Supported in part by NSF grant CCF-1816372. }  \and Ryan O'Donnell\thanks{{\tt CMU}. {\tt odonnell@cs.cmu.edu}. Supported by NSF grant FET-1909310. } \and Pedro Paredes\thanks{{\tt CMU}. {\tt preisben@cs.cmu.edu}. Supported by NSF grant FET-1909310. } \and Madhur Tulsiani\thanks{{\tt TTIC}. {\tt madhurt@ttic.edu}. Supported by NSF grant CCF-1816372.}}

\date{\today}



\maketitle
\thispagestyle{empty}

For an abelian group $H$ acting on the set $[\ell]$, an
$(H,\ell)$-lift of a graph $G_0$ is a graph obtained by replacing each
vertex by $\ell$ copies, and each edge by a matching corresponding to
the action of an element of $H$.

In this work, we show the following \emph{explicit} constructions of
expanders obtained via abelian lifts. For every (transitive) abelian
group $H \leqslant \Sym(\ell)$, constant degree $d \ge 3$ and
$\epsilon > 0$, we construct explicit $d$-regular expander graphs $G$
obtained from an $(H,\ell)$-lift of a (suitable) base $n$-vertex
expander $G_0$ with the following parameters:
\begin{itemize}
  \item[(i)] $\lambda(G) \le 2\sqrt{d-1} + \epsilon$, for any lift size $\ell \le 2^{n^{\delta}}$ where $\delta=\delta(d,\epsilon)$,
  \item[(ii)] $\lambda(G) \le \epsilon \cdot d$, for any lift size $\ell \le 2^{n^{\delta_0}}$ for a fixed $\delta_0 > 0$, when $d \ge d_0(\epsilon)$, or
  \item[(iii)] $\lambda(G) \le \widetilde{O}(\sqrt{d})$, for lift size ``exactly'' $\ell = 2^{\Theta(n)}$.
\end{itemize}
As corollaries, we obtain \emph{explicit} quantum lifted product codes
of Panteleev and Kalachev of almost linear distance (and also in a
wide range of parameters) and \emph{explicit} classical quasi-cyclic
LDPC codes with wide range of circulant sizes.

Items $(i)$ and $(ii)$ above are obtained by extending the techniques
of Mohanty,  O'Donnell and Paredes [STOC 2020] for $2$-lifts to much
larger abelian lift sizes (as a byproduct simplifying their
construction). This is done by providing a new encoding of special
walks arising in the trace power method, carefully ``compressing''
depth-first search traversals. Result $(iii)$ is via a simpler proof
of Agarwal \etal [SIAM J. Discrete Math 2019] at the expense of
polylog factors in the expansion.

\newpage

\ifnotconf
\pagenumbering{roman}
\tableofcontents
\clearpage
\fi

\pagenumbering{arabic}
\setcounter{page}{1}

\section{Introduction}
Graphs are ubiquitous in theoretical computer science and the ability
to explicitly construct graphs with special properties can be quite
useful. Two such properties are expansion and symmetry. A graph is
expanding if it is simultaneously sparse and highly connected (meaning
that we need to remove a lot of edges to disconnect a large part of
the graph.) The theory of
\emph{explicit} constructions of expander graphs has seen a dramatic
development over the past four decades\footnote{See~\cite{HooryLW06}
  for an excellent survey on expander graphs.}
\cite{LPS88,Margulis88,Morgenstern94,RVW00,BL06,BT08,Cohen16,MOP20,
  OW20, Alon21}. We now have constructions via diverse methods
achieving a wide range of expansion guarantees. These range from very
explicit algebraic constructions of so-called Ramanujan
graphs~\cite{LPS88} to recursive combinatorial ones based on the
Zig-Zag product~\cite{RVW00}. These constructions have a plethora of
applications specially to coding theory and
pseudorandomness~\cite{Vadhan12}.  A highly sought goal is to make the
expansion of a family of (bounded) degree $d$ graphs as close to the
Ramanujan bound as possible, \ie having largest non-trivial eigenvalue
at most $2\sqrt{d-1}$. The Alon-Boppana bound\cite{Nil91} states that
the largest non-trivial eigenvalue is at least $2\sqrt{d - 1} - o(1)$,
so the Ramanujan bound is in a sense optimal. This goal of achieving
strong spectral guarantees has been an important motivation.

Moving beyond spectral guarantees, we can ask for graphs that combine
the important property of expansion with additional structure and the
one we focus on is \textit{symmetry}\footnote{Informally, we say that
$G$ has symmetries of $H$ if $H \subseteq \aut(G)$, where $\aut(G)$
denotes the group of all graph isomorphisms to itself.  }.
One of the problems that has been studied in graph theory is to
construct graphs with a given automorphism group. Frucht proved in
1939 that for every finite group $H$, we have a graph $G$ such that
$\aut(G) = H$ . Babai \cite{B74} later showed that there is such a
graph on at most $2|H|$ vertices\footnote{Except for $\Z_3$, $\Z_4$
and $\Z_5$.}. Thus, we have a natural question

\begin{center}
  \emph{Can we explicitly construct expanding graphs with given symmetries? }
\end{center}

While interesting in its own right, the ability to control symmetries
also has concrete applications. For example, a very recent work
\cite{GW21} constructs many families of expanding asymmetric graphs,
\ie having no symmetries, and shows applications to property testing and
other areas. We will focus on an important connection to both quantum
and classical codes that was the motivation behind this work.

Low-density parity check (LDPC) codes were first introduced by
Gallager \cite{Gal62} in the '60s and are one of the most popular
classes of classical error-correcting codes, both in theory and in
practice.  LDPC codes are linear codes whose parity check matrices
have row and column weights bounded by a constant (which means that
each parity check depends only on a constant number of bits).  The
popularity of this family of codes comes from the fact that there are
many known constructions of classical LDPC codes that achieve linear
rate and distance that can also be decoded in linear time \cite{RU08}.

A family of codes that has been extensively studied is cyclic
codes, \ie codes that are invariant under the action of $\Z_N$ where
$N$ is the blocklength. This symmetry leads to efficient encoding and
decoding algorithms and a major open problem is whether good cyclic
codes exist. Babai, Shpilka and Stefankovich \cite{BSS05} showed that
cyclic codes cannot be good LDPC codes and this negative result was
extended by Kaufman and Wigderson \cite{KW10} to LDPC codes with
a \textit{transitive} action by an arbitrary abelian group.

Quasi-cyclic codes are a generalization of cyclic codes in which
symmetry is only under rotations of multiples of a parameter (called
index) $n$ where $N = n\ell$. This is equivalent to relaxing the
transitivity condition to allow for $n$ orbits. Unlike cyclic codes,
good quasi-cyclic codes are known to exist as was shown by Chen,
Peterson and Weldon \cite{CPW69}. More recently, Bazzi and
Mitter \cite{BM06} gave a randomized construction for any constant
$n>2$ and showed that it attains Gilbert--Varshamov bound rate $1/n$.
Quasi-cyclic codes have been extensively studied and are very useful
in practice (\eg their LDPC counterparts are part of the 5G standard
of mobile communication \cite{LBMZX18}).

In the realm of quantum computing, the fragility of qubits
makes quantum error correcting codes crucial for the
realization of scalable quantum
computation. \emph{Calderbank-Shor-Steane} (CSS) codes are a family of
quantum error-correcting codes that was first described
in \cite{CS96, Ste96}. A CSS code is defined by a pair of classical linear
codes that satisfy an orthogonality condition. The quantum analog of
LDPC codes is thus defined as CSS codes where the parity check
matrices of both codes have bounded row and column weights.

Constructing quantum LDPC codes of large distance has been active area
of research recently. After two decades, \cite{EKZ20} broke the
$\sqrt{N}$ barrier and there was a flurry of activity
with \cite{hho20} extending it to $N^{3/5}$ (up to $\poly\log$
factors). Panteleev and Kalachev \cite{PK20} came up with another
breakthrough construction achieving almost linear
distance. Both \cite{hho20} and \cite{PK20} are non-explicit
constructions crucially relying on symmetries. The construction
in~\cite{PK20} interestingly used quasi-cyclic LDPC codes which in
turn was constructed using expander graphs with cyclic symmetry.
Moreover, Breuckmann and Eberhardt~\cite{BreuckmannE21} introduced a
new approach for constructing quantum codes simultaneously
generalizing \cite{hho20} and \cite{PK20} in order to obtain explicit
codes out of a pair of graphs having the symmetries of any group.
This provides a very concrete motive to study explicit construction of
expander graphs symmetric under various families of groups. \vspace{-
11 pt}

\paragraph*{Update} \hspace{-0.4cm}
While the current paper was being prepared for publication, Panteleev
and Kalachev \cite{PK21} found a breakthrough construction of explicit
good quantum LDPC codes.

\vspace{- 5 pt}
\subsection*{Current Techniques}
Many of the current known constructions of expanders are Cayley graphs
and therefore are highly symmetric but are somewhat rigid in the sense
that one may not be able to finely control the symmetries of a given
construction. One general approach is to construct an expanding Cayley
graph for a given group but the Alon--Roichman theorem \cite{AR94} only
guarantees a logarithmic degree which is tight when the group
is abelian (and this large degree is undesirable for some application
in coding theory). The other technique used to build expanders is via
an operation called \emph{lifting}.

In general form, the \emph{random} lifting operation takes a lift size
parameter $\ell$, a base expander graph $G_0$ on $n$ vertices and a
subgroup $H$ of the symmetric group, $\Sym(\ell)$, and constructs a new
``lifted'' graph $G$ on $n\ell$ vertices where each vertex $v$ of
$G_0$ is replaced by $\ell$-copies $(v,1),\ldots,(v,\ell)$ and for
every edge $e=(u,v)$ of $G_0$ a uniformly random element of $h_e \in
H$ is sampled and $(u,i)$ is connected to $(v, h_e(i))$ for $i \in
[\ell]$. We say that $G$ obtained this way is a \emph{random}
$(H,\ell)$-lift of $G_0$. We call it an \emph{unstructured}
$\ell$-lift if there is no restriction on the group, \ie $H
= \Sym(\ell)$.

Lifting has three very useful properties. One, it preserves the degree
of the base graph. Secondly, random lifts preserve
expansion\footnote{This holds for any lift size in the case of
``unstructured'' $\ell$-lifts, but only holds for $\ell \le
2^{O_d(n)}$ when $H$ is abelian (and transitive).} with high
probability. Finally, (and importantly for us), if $H$ is abelian,
then the lifted graph inherits symmetries of $H$. The first two
properties are clearly useful in constructing larger expanders from a
small one, and for this reason, there has been extensive work on lift
based constructions.

Bilu and Linial~\cite{BL06} introduced \emph{2-lifts} in an explicit
construction of graphs with expansion $O(\sqrt{d}\log^{1.5}(d))$ for
every degree. More recently, Mohanty, O'Donnell and
Paredes~\cite{MOP20} gave the first explicit construction of
near-Ramanujan, \ie largest non-trivial eigenvalue bounded by
$2\sqrt{d-1} + \epsilon$, graphs of every degree. The key technique in
their work was a finer derandomization of $2$-lifts. Subsequently,
Alon~\cite{Alon21} gave explicit constructions of near-Ramanujan
expanders of every degree and every number of vertices. The work
in~\cite{MOP20} was also generalized to achieve finer spectral
guarantees together with local properties via \emph{unstructured}
$\ell$-lifts in O'Donnell and Wu~\cite{OW20}.

When one restricts $H$ to be abelian, Agarwal \etal~\cite{ACKM19}
showed that \emph{random} $(\mathbb{Z}_\ell,\ell)$-lifts (also known
as \emph{shift lifts}) are expanding. Motivated by the applications of
these lifts to codes, we obtain explicit constructions of expanding
abelian lifts, for a wide range of lift sizes.

\subsection{Our Results and Techniques}

Our construction of the lifts (and the expansion thereof) vary based
on the parameter $\ell$ and we make the following classification for
ease in presenting the results. Let $n,d,\varepsilon$ be given.

\begin{itemize}
\item \textit{Sub-Exponential} - This is the regime where $\ell \le \exp\left(  n^{\delta(d,\varepsilon)} \right)$. The exponent $\delta(d,\varepsilon)$ goes to zero as the degree ($d$) increases or  $\varepsilon$ vanishes.
\item \textit{Moderately-Exponential} - This is when $\ell \le \exp\left(  n^{\delta_0} \right)$. The exponent is some fixed universal constant $\delta_0 \in (0,1)$ .
\item \textit{Exactly-Exponential} - This is the regime where $\ell = \exp(\Theta_d(n))$.
\end{itemize}

Our first main result shows explicit constructions in the
sub-exponential and moderately exponential regimes.

\begin{restatable}{theorem}{MainI}\label{theo:main1}
  For large enough $n$ and constant degree $d \ge 3$, given $\ell$ such that $\ell \leq \exp(n^{\Theta(1)})$,
  the generating elements of a transitive abelian group $H \leqslant \Sym(\ell)$, and any fixed constant $\varepsilon \in (0,1)$,
  we can construct in deterministic polynomial time, a $d$-regular graph $G$ on $\Theta(n\ell)$ vertices such that
\begin{itemize}
    \item $G$ is $(H,\ell)$-lift of a graph $G_0$ on $\Theta(n)$ vertices.
    \item{(Sub-Exponential)} If $\ell \le \exp\left(  n^{\delta(d,\varepsilon)} \right)$, then $\lambda(G) \le 2\sqrt{d-1} + \epsilon$.
    \item{(Moderately-Exponential)} If $\ell \le \exp\left( n^{\delta} \right)$ and also $d \ge d_0(\epsilon)$, then $\lambda(G) \le \varepsilon \cdot d$.
\end{itemize}
\end{restatable}

The bulk of the technical work is in the proof of \cref{theo:main1}.
For this, we build on the techniques of~\cite{MOP20} for derandomizing
$2$-lifts via the trace power method. When analyzing larger lift sizes
(required in our derandomization of quantum and classical codes), we
are led to consider much larger walk lengths in the trace power
method. A central technical component in their work is the counting of
some special walks which ultimately governs the final spectral bound
of the construction. For lift sizes larger than
$2^{2^{\Theta(\sqrt{\log n})}}$, their counting trivializes no longer
implying expansion of the construction.  Our main technical
contribution consists in providing alternative ways of counting such
special walks by carefully compressing the traversal of the
depth-first search (DFS) algorithm.

We are able to extend the near-Ramanujan guarantee for $2$-lifts
from~\cite{MOP20} to the entire sub-exponential regime of lift sizes
$\ell$. In the moderately exponential regime, the walks are too long
and we resort to another counting that can only guarantee an expansion
of $\varepsilon \cdot d$. \cref{theo:main1} can be seen (slight)
simplification of the construction in~\cite{MOP20} since we can now do
a single large lift instead of performing a sequence of $2$-lifts as
in their work\footnote{Performing a single lift also has the advantage
of having to meet a technical condition (bicycle-freeness) only once
instead of at each lift operation.}.

Let us now formally state the results of Agarwal \etal
in~\cref{theo:agarwal_et_al} showing \emph{randomized} constructions
of abelian lifts.

\begin{theorem}[Agarwal \etal~\cite{ACKM19}, Theorem 1.2]\label{theo:agarwal_et_al}
  Let $G_0$ be a $d$-regular $n$-vertex graph, where $2 \le d \le \sqrt{n/(3\ln{n})}$.
  Let $G$ be a \emph{random} $(\mathbb{Z}_\ell,\ell)$-lift of $G_0$. Then
  \begin{align*}
     \lambda(G) = O(\lambda(G_0)),
  \end{align*}
  with probability $1 - \ell \cdot e^{-\Omega(n/d^2)}$. Moreover, if $\ell \geq \exp( O_\varepsilon(nd) )$, then no abelian $(H, \ell)$ lift has $\lambda(G) \le \varepsilon \cdot d$.
\end{theorem}

This result is based on discrepancy methods building on the work of
Bilu and Linial \cite{BL06} and gives lower and upper bounds that are
tight up to a factor of $d^3$ in the exponent.

\cref{theo:main1} can be seen as a (derandomization of the parameters) in \cref{theo:agarwal_et_al}
for every constant degree and lift size from $2$ all the way to
$\exp(n^{\Theta_d(1)})$. In the sub-exponential regime, our result
improves their spectral guarantee from $O(\sqrt{d})$ to $2\sqrt{d-1}
+ \epsilon$.

Our second main result shows explicit constructions in the exponential
regime. While it is not hard to observe that one can derandomize the
exponential lift by using off-the-shelf tools, we give a short proof
via a key lemma of Bilu and Linial~\cite{BL06} that is a converse of
the expander mixing lemma. Although it gives a spectral guarantee that
is weaker by a log factor, it yields an accessible proof and moreover,
interpolates the exponent from $\exp(O(n/d^2))$ all the way to the
barrier of $\exp(O(nd))$ thereby bridging the $d^3$-gap.

\begin{theorem}[Exactly Exponential Lifts]\label{theo:main2}
  For any positive integers $n,\ell$ and every constant degree $d \ge 3$, given $\ell$,
  the generating elements of a transitive abelian group $H \leqslant \Sym(\ell)$, there exists a deterministic $\poly(\exp(n),\ell)$
  time algorithm that constructs a $d$-regular graph $G$ on $n\ell$ vertices such that
  \begin{itemize}
    \item $G$ is $(H,\ell)$-lift of a graph $G_0$ on $n$ vertices, and
    \item  If $\ell \leq \exp\left(\Theta\left(n/\sqrt{d}\right)\right)$, then $\lambda(G) \le O\left( \sqrt{d} \cdot \log d\right)$.
    \item  If $\ell = \exp\left(\Theta\left(nd^\delta\right) \right)$ for $\delta \in [-1/2, 1)$, then $\lambda(G) \le O\left(d^{\frac{2+\delta}{3}} \cdot \log d\right)$.
  \end{itemize}
  In particular, we have explicit polynomial time construction of a lift when $\ell = \exp(\Theta(n))$.
\end{theorem}

\subsection{Derandomized Quantum and Classical Codes}

We first state the code constructions in~\cite{PK20} and then show how
large explicit abelian lifts derandomize their codes.

\begin{theorem}[\cite{PK20}]\label{theo:pk}
   Let $G$ be a $d$-regular graph on $n\ell$-vertices such that $G$ has a symmetry\footnote{To be more precise, $\Z_\ell$ acts freely on $G$.}
   of $\Z_\ell$ and $\lambda_2(G) \leq \varepsilon \cdot d$. Then we can construct the following,
   \begin{itemize}
       \item A good quasi-cyclic LDPC code of block length $N=\Theta(n\ell)$ and index $\Theta(n)$.
       \item A quantum LDPC code which has distance $\Theta_{\epsilon,d}(\ell)$ and dimension $\Theta(n)$.
   \end{itemize}
\end{theorem}

Panteleev and Kalachev use the aforementioned \emph{randomized}
construction of abelian lifted expanders by Agarwal
\etal~\cite{ACKM19}, where each edge of the base graph is a associated
with an element in $\mathbb{Z}_\ell$ sampled uniformly. When $\ell$ is
in the exponential regime they obtain quantum LDPC codes with almost
linear distance, \ie $\Omega(N/\log(N))$.

Breuckmann and Eberhardt~\cite{BreuckmannE21} gave a derandomization
of \cite{PK20} in a more restricted parameter regime by observing that
the Ramanujan graph construction by Lubotsky, Philips and
Sarnak \cite{LPS88} of size $n$ has a (free) action of
$\Z_{n^{1/3}}.$ By \cref{theo:pk}, we have an explicit quantum LDPC
code of distance $O(N^{1/3})$ under the notion of
distance\footnote{\cite{BreuckmannE21} state their result for a
slightly different notion of a quantum codes called subsystems codes
for which the corresponding distance (also known as dressed distance)
is larger.}  in~\cite{PK20,hho20}.

As a direct corollary of \cref{theo:main2}, we have a complete
derandomization of \cite{PK20} yielding explicit quantum LDPC codes of
almost linear distance. This greatly improves the distance of the
existing explicit construction. We also get good quasi-cyclic LDPC
codes of almost linear circulant size. Moreover, the ability to
construct a wide range of lift sizes from \cref{theo:main1} lets us
control the circulant size which can be useful in practice. By
controlling the lift size, we can also directly amplify the rate of
their quantum LDPC codes (without resorting to the product of
complexes). To summarize,

\begin{corollary}[\cite{PK20}, \cref{theo:main1} ,\cref{theo:main2}]
   We have explicit polynomial time construction of each of the following,
   \begin{itemize}
       \item Good quasi-cyclic LDPC code of block length $N$ and
  any circulant size up to $N/\polylog(N)$ or $\Theta(N/\log(N))$.
       \item Quantum LDPC code with distance $\Omega(N/\log(N))$
  and dimension $\Omega(\log(N))$.
        \item Quantum LDPC code with distance $\Omega(N^{1-\alpha})$
  and dimension $\Theta(N^{\alpha})$ for every constant $\alpha > 0$.
   \end{itemize}
\end{corollary}

\subsection*{Further Directions}
Our work also leads to several natural avenues for further exploration.
\begin{enumerate}
\item \textit{More Symmetries} - While these lift-based constructions yield graphs with symmetries arising from abelian groups, it is interesting to understand whether one can construct sparse graphs with symmetries corresponding to other families of groups.
Such constructions may require new ways of using the symmetry groups, in ways other than in lifts of a base graph.
More generally, it may be useful to investigate other ways of exploiting graph symmetry, beyond their applications to codes.
\item \textit{Better notions of explicitness} - It is a very interesting problem to find strongly explicit constructions of lifted abelian expander. Even making the running time closer to linear would be interesting.
Also, since quasicyclic codes are widely used in practice, it may be helpful to find explicit constructions which are efficiently implementable.
\item \textit{Complete Range} - Can we derandomize abelian lifts for $\ell$ in between $2^{n^{\Theta(1)}}$ and $2^{O_d(n)}$? Can we extend the near-Ramanujan bound beyond the subexponential range?
\end{enumerate}

\section{Preliminaries}
For an operator $M$, let its eigenvalues be ordered such that
$\{ \abs{\lambda_1(M)} \ge \cdots \ge \abs{\lambda_n(M)}\}$. We define
$\rho_2(M) = \abs{\lambda_2(M)}$. For an an $n$-vertex graph $G =
(V,E)$, we denote by $\lambda(G) = \rho_2(A)$, where $A$ is its
adjacency operator.

We assume that we have an an ordering on $V$ and by convention, $(u,v) \in E$ if $u \leq v$.

A \emph{character}\footnote{The definition we give is that of
a \emph{linear character}. We use the term character as we work only
with abelian groups. } of a group is a map $\chi: H \to \C^*$ that
respects group multiplication, i.e., $\chi(h_1h_2)
= \chi(h_1)\chi(h_2)$. For a finite group $|\chi(h)| = 1$ for every
$h \in H$. The \textit{trivial character} is the one which has
$\chi(h) = 1$ for every $h$. The rest of the characters we
call \textit{non-trivial}.

The action of a group $H$ on a set of $\ell$ elements is defined by a
map $\psi: H \to \Sym(\ell)$ which satisfies $\psi(h_1h_2)
= \psi(h_1)\psi(h_2) $.  Since we only care about the action of the
group, we will assume that our input is actually
$\psi(H) \subseteq \Sym(\ell)$ and the action is the natural one.

\begin{definition}[$(H,\ell)$-lift of a graph]
    An $(H,\ell)$-signing of an undirected graph $G = (V,E)$ is a
    function $s: E \to H \subseteq \Sym(\ell)$. The lifted graph $G(s)
    = (V',E')$ is a graph on $\ell$ copies of the vertices $V' = V\times
    [\ell]$ where for every edge $(u,v) \in E$ we have $((u,i),
    (v,s(u,v)\acts i) ) \in E'$
\end{definition}

We will restrict to analyzing abelian $H$ and the most important case
to consider is when $H = \Z_\ell$, i.e. the cyclic group. A necessary
condition for the lift to be expanding is for it to be connected. A
subgroup $H$ is \textit{transitive} if for every $i, j \in [\ell]$,
there exists $h \in H$ such that $h\cdot i = j$. Lifts of
non-transitive subgroups are disconnected because if the pair
$\{i,j\}$ violate the condition then any pair $(u,i)$ and $(v,j)$ are
disconnected. Thus, we will assume henceforth that we work with
transitive abelian subgroups.

Let $E^d$ denote the set of directed edges i.e. $E^d = \{(u,v),
(v,u) \;|\; (u,v) \in E \} $. We extend the signing to $E^d$ such that
for an edge $(u,v) \in E $, $\;s(v,u) := s(u,v)^{-1}$.

\medskip
\begin{definition}[Non-backtracking walk operator]\label{def:backtrack}
  For an extended signing $s:E^d\to H$ and a character $\chi$ of $H$,
  the signed non-backtracking walk matrix $B_s(\chi)$ is a
  non-symmetric matrix of size $|E^d| \times |E^d|$ in which the entry
  corresponding to the pair of edges $(u,v), (x,y)$ is $\chi(s(x,y))$
  if $v = x, \; u \neq y$, and zero otherwise.
\end{definition}

The unsigned variant is obtained by taking the trivial character in
the definition above. Let the non-backtracking walk matrix of $G$ be
$B$ and the lifted graph with respect to a signing $s$ be
$B_{G(s)}$. We use the following standard facts.

\begin{fact}\label{fact:ihara}
Let $B$ be the non-backtracking walk matrix of a $d$-regular graph $G$. 
Then, 
\[
\lambda(G) ~\leq~  2\cdot \max\{ \sqrt{d-1},\; \rho_2(B)\}.
\]
\end{fact}

\begin{fact}\label{fact:character_decomp}
  If $H\subseteq \Sym(\ell)$ is abelian, then there exist characters
  $\{\chi_1,\cdots ,\chi_\ell\}$\footnote{These need not be
  distinct. For example if $H$ is trivial, then all the $\chi_i$ are
  trivial} such that we have $\spec(B_{G(s)})
  = \bigcup_i \spec(B_s(\chi_i))$. If $H$ is transitive, then exactly
  one of the characters is trivial.
\end{fact}

\section{Proof Strategy}
We give an overview of the proof of \cref{theo:main1}.  As mentioned
earlier, our results build on the work of Mohanty, O'Donnell and
Paredes~\cite{MOP20}, so we briefly recall notions and ideas from
their work that we will need.

Let $G_0$ be a base expander graph and $s \colon E_0 \to \mathbb{Z}_2$
be a signing that defines a lift. It is convenient to first think that
the signing is chosen uniformly at random and later see which
properties were indeed used so that an appropriate derandomization
tool may be used. Using well known facts (\cref{fact:ihara}
and \cref{fact:character_decomp}) they reduce the problem of analyzing
the expansion of the lifted graph to that of bounding the spectral
radius $\rho(B_s)$ of the non-backtracking operator $B_s$.

\noindent \textbf{The MOP Argument:} A common technique to bound the spectral radius 
is the trace power method which in our case amounts to counting
special non-backtracking walks. This is the motivation for using the
non-backtracking operator $B_s$ instead of the more common adjacency
operator which require counting closed walks (which is potentially
harder). Another standard fact\footnote{To avoid discussing some
unimportant technicalities, we will make some simplifications in this
high-level overview.} is that
\begin{align*}
  \rho(B_s)^{2k} ~\le~ \tr((B_s^{*})^kB_s^k) ~=~ \sum_{\substack{(e_1,\ldots,e_{2k}) \\ \textup{closed edge walk}}} \prod_{i=1}^{2k} \chi(s(e_i)).
\end{align*}

The above expression greatly simplifies when we take the expectation
over a uniformly random signing since only walks in which every edge
occurs at least twice stand a chance of surviving the
expectation. These walks are called singleton free in~\cite{MOP20}. We have
\begin{align*}
  \E_{s \in \mathbb{Z}_2^{E_0}}\left[ \rho(B_s)^{2k} \right] ~\le~ \sum_{\substack{(e_1,\ldots,e_{2k}) \\ \textup{closed edge walk}}} \E_{s \in \mathbb{Z}_2^{E_0}}\left[\prod_{i=1}^{2k} \chi(s(e_i)) \right] \le \abs{\left\{\substack{\textup{$2k$-length singleton free} \\ \textup{non-backtracking walks in $G_0$}}\right\}},
\end{align*}
reducing the problem of bounding the spectral radius to a counting
problem of these special walks. In the hypothetical (idealized)
scenario of $G_0$ being Ramanujan and the counting on the RHS above
being $(d-1)^{k}$, we would have a Ramanujan lift. The above
expression also hints that $\epsilon$-bias distributions might be a
useful derandomization tool here. This idealized scenario can be too
optimistic and the count of $(d-1)^{k}$ has additional factors, but
they remain small after taking a $2k$-th root (when $k$ is neither too
small or large)

One of the main technical contributions in~\cite{MOP20} is the
counting of $2k$-length singleton free non-backtracking walks in
$G_0$, which they call hikes.  For the sake of intuition, we will
assume that $G_0$ has girth $g$, but it is not hard to modify the
argument when $G_0$ has at most one cycle around any neighborhood of
radius $< g/2$ centered at vertex in $G_0$ (the bicycle freeness
property). They view the vertices and edges visited in a hike as
forming a hike graph $\mathcal{H}$. Assuming that $g
= \Omega(\log_{d-1}(n))$, if $k$ is not too large, then $\mathcal{H}$
looks like a tree possibly with a few additional edges forming cycles
as established by Alon, Hoory and Linial in~\cite{AHL02} (and
generalized in \cite{MOP20} to bicycle-free radius from girth).

Assuming that the hike is singleton free, we can have at most $k$
steps that visit an edge that was not previously visited. This implies
that the hike graph $\mathcal{H}$ has at most $k$ edges and at most
$k+1$ vertices (since it is connected).  They count the number of
these special walks by directly specifying an encoding for the hike.
Up to negligible factors (after $2k$-th root for $k$ not too small),
they show that there are at most
\begin{align*}
  n \cdot (d-1)^{k} \cdot k^{O\left(\frac{\ln(k)}{g}\right)  \cdot k},
\end{align*}
singleton free hikes of length $2k$ (see \cite[Theorem 3.9]{MOP20} for
precise details). This bound trivializes, \ie it becomes at least
$(d-1)^{2k}$, for $\ln(k) \gg \sqrt{g}
= \Theta\left(\sqrt{\log_{d-1}(n)}\right)$. This means that we cannot
use their bound for very long walks and this in turn prevents us from
getting lift sizes larger than $2^{2^{\Theta(\sqrt{\log_{d-1}(n)})}}$
from their results.

\noindent \textbf{Our Approach:} Now, let's consider $\mathbb{Z}_\ell$
lifts for large $\ell$. The spectral radius of each individual
$B_s(\chi)$ can be analyzed in a similar fashion as above via the
trace power method. However, we need to bound all of
them \emph{simultaneously}. We know no better way than a simple union
bound over the $\ell-1$ cases, but this will force us to obtain a much
better concentration guarantee out of the trace power method which in
turn entails having to consider much larger walk lengths.

Instead of encoding a hike directly as in~\cite{MOP20}, we will first
encode the subgraph of $G_0$ traversed by the hike, which we call hike
graph, and then encode the hike having the full hike graph at our
disposal. We will give two different encodings for the hike graph. The
first one is simpler and can encode an arbitrary graph. The second
encoding uses the special structure of the hike graph, namely, having
few vertices of degree greater than $2$. Both encodings are based on
the traversal history of the simple depth-first search (DFS)
algorithm. Let $\mathcal{H}$ be the hike graph on $m \le k$ edges and
$n' \le k+1$ vertices. As DFS traverses $\mathcal{H}$, each of its
edges will be visited twice: first ``forward'' via a recursive call
and later ``backwards'' via a backtracking operation.  We view each
step of the DFS traversal as being associated with an edge that is
being currently traversed and the associated type of traversal:
recursive (R) or backtracking (B). A key observation is that only for
the recursive traversals we need to know the next neighbor out of
$d-1$ possibilities (except for the first step). For the backtracking
steps, we can rely on the current stack of DFS. Thus, if we are given
a starting vertex from $G_0$, a binary string in $\set{R,B}^{2m}$ and
a next neighbor for each recursive step, we can reconstruct
$\mathcal{H}$. Note that there at most
\begin{align*}
  n \cdot d \cdot (d-1)^{k} \cdot 2^{2k},
\end{align*}
such encodings. Having access to the hike graph and again assuming
that the graph has girth $g = \Omega(\log_{d-1}(n))$ (similarly,
bicycle freeness is also enough). Using the locally tree-like
structure, a $2k$-length hike can be specified by splitting it into
segments of length $< g/2$, by specifying the starting vertex of the
first segment and the ending vertex of each segment, we have enough
information to recover the full hike. Note that there are at most
\begin{align*}
  k^{O(k/g)},
\end{align*}
ways of encoding a hike. Then, the number of $2k$-hikes in $G_0$ is at
most
\begin{align*}
  n \cdot d \cdot (d-1)^{k} \cdot 2^{2k} \cdot k^{O(k/g)}.
\end{align*}
Now we can take $k \approx n^{\delta}$ for a sufficiently small
$\delta=\delta(d) > 0$ and obtain, after taking the $2k$-th root of
the above quantity,
\begin{align*}
  \rho(B_s) ~\le~ (1+\epsilon) \cdot 2 \cdot \sqrt{(d-1)},
\end{align*}
when $k=k(n,d,\epsilon)$ is sufficiently large and $c=c(\epsilon)$ is
sufficiently small. The extra factor $2$ prevent us from obtaining
near-Ramanujan bounds with this counting. Nonetheless, the simple
counting already allows us to obtain expansion $O(\sqrt{d})$ for lifts
sizes as large as $2^{n^{\delta(d)}}$. Moreover, by weakening the
expansion guarantee we can obtain lift sizes as large as
$2^{n^{\Theta(1)}}$ from this counting and obtain part
of~\cref{theo:main1}. If we insist on getting a near-Ramanujan bound,
we need to compress the traversal history further since storing a
string $\set{R,B}^{2m}$ is too costly and leads to this factor of
$2$. Note that this string has an equal number of $R$ and $B$ symbols,
so it cannot be naively compressed.

To obtain a near-Ramanujan graph, we will take advantage of the
special structure of the hike graph (when the walk length is large but
not too large) in which most of its vertices have degree exactly $2$.
These degree $2$ vertices are particularly simple to handle in a DFS
traversal. For them, we only need to store the next neighbor out of
$d-1$ possibilities in $G_0$ (except possibly for the first step). In
a sequence of backtrackings, if the top of the DFS stack is a degree
$2$ vertex we know that we are done processing it since no further
recursive call will be initiated from it. Then, we simply pop it from
the stack. It is for the ``rare'' at most $\delta \cdot n'$ vertices
$v$ of degree $\ge 3$ that we need to store how many extra recursive
calls $t_v$ we issue from $v$ and a tuple of additional next neighbors
$(d_1,\dots,d_{t_v})$. The total number of such encodings is at most
\begin{align*}
  n \cdot d \cdot (d-1)^{k} \cdot \binom{k+1}{\delta (k+1)} \cdot (d-1)^{\delta (k+1)},
\end{align*}
which combined with the same previous way of encoding a hike given its graph
results in a total number of hike encodings of $G_0$ of at most
\begin{align*}
  n \cdot d \cdot (d-1)^{k} \cdot \binom{k+1}{\delta (k+1)} \cdot (d-1)^{\delta (k+1)} \cdot k^{O(k/g)},
\end{align*}
By choosing $\delta=\delta(d,\epsilon)$ sufficiently small and taking
$k=k(n,d,\epsilon) \le 2^{\delta \cdot g} \approx n^{O_d(\delta)}$ sufficiently large, we obtain
after taking the $2k$-th root
\begin{align*}
  \rho(B_s) ~\le~  \sqrt{(d-1)} + \epsilon,
\end{align*}
indeed leading to a near-Ramanujan bound for lifts as large as
$2^{n^{\delta}}$ in \cref{theo:main1}.

Now we briefly explain how to handle the union bound to ensure that
$\rho(B_s(\chi))$ is \emph{simultaneously} small for all $(\ell-1)$
non-trivial characters (in the decomposition
of \cref{fact:character_decomp}). This union bound is \emph{standard}
when using the trace power method, what is relevant is the trade-off
between lift size and walk length. To obtain a high probability
guarantee from a guarantee on expectation, it is standard to consider
larger walk lengths from which concentration follows from a simple
Markov inequality. More precisely, if for some function $f\;$,
$\E \rho(B_s(\chi_{j}))^{2k} \le f(n,d,g,k)$, then by Markov's
inequality,
\begin{align*}
  \Pr_{s \in \mathbb{Z}_{\ell}^{E_0}} \left[\rho(B_s(\chi)) \ge 2^{\log_2(\ell)/(2k)} \cdot f(n,d,g,k)^{1/(2k)} \right] ~\le~ \frac{1}{\ell}.
\end{align*}
Therefore, for $k \ge \log_2(\ell)$ sufficiently large, we can union
bound over all characters $\chi$ and obtain similar bounds as
before. As alluded above, this lower bound on the length of the walk
depending on the lift size is the reason why we are led to consider
much longer walks. To conclude this proof sketch, we need to replace a
random signing by a pseudorandom random one. As in~\cite{MOP20}, we
use $\epsilon$-biased distributions but suitably generalized to
abelian groups, \eg the one\footnote{For our application, it suffices
to have the support size of the $\epsilon$-biased distribution
polynomial in $1/\epsilon$.} by Jalan and Moshkovitz
in \cite{JM21}. We may be taking very large walks on the base graph
$G_0$, so the error of the generator needs to be smaller than $n \cdot
d^{2k}$, where $k$ can be as large as $n^{\Theta(1)}$. We note that as
long as the degree $d$ is a constant this quantity is at most a
polynomial in the size of the \emph{final} lifted graph $G$ since
walks of length $O(\log(\abs{V(G)}))$ suffice for any lift size up to
full extent of $2^{O(n)}$, for which abelian lifts can be expanding.

The above argument covers~\cref{theo:main1}, namely, the
sub-exponential and moderately-exponential abelian lift sizes. The
``exact'' exponential regime of~\cref{theo:main2} relies on an elegant
converse of the expander mixing lemma by Bilu and
Linial~\cite{BL06}. Since this regime is simpler, we defer the details
to~\cref{sec:exact_exponential_lift}, where it is formally presented.

\section{A New Encoding for Special Walks}
In this section we will count the total number of singleton-free hikes
of a given length on a fixed graph, $G$. We split the count into two
parts. First, we count the number of possible hike graphs and then,
for a given hike graph $\cH$, we count the number of hikes that
can \ie yield $\cH$ on traversal. Each of these counts is via an
encoding argument and therefore we have two kinds of encoding. One for
graphs and the other for hikes. In the first part of the section we
give two ways of encoding graphs, and in the other half, we encode
hikes. Since the first section is a general encoding for subgraphs, we
relegate formal definitions related to hikes to a later section.

\subsection{Graph Encoding}
Let $\cH$ be a subgraph of a fixed $d$-regular graph $G$. We wish to
encode $\cH$ in a succinct way such that given the encoding and $G$,
we can recover $\cH$ uniquely. We will give two ways of encoding
$\cH$. The first one will be generic that works for any subgraph of a
$d$-regular graph. The second encoding takes advantage of the special
sparse structure (not too many vertices of degree greater than
two). We assume that we have an order on the neighbors of every
vertex, and thus, given $(v,j)$, we can access the $j^{th}$ neighbor
of $v$ efficiently.

We will do this by encoding a DFS based-traversal of it from a given
start vertex .  Here, we really need our DFS traversal to be optimal
in the sense that the number of times each edge is traversed is at
most two and not any higher. We, therefore, include precise details of
our implementation in ~\cref{sec:dfs_implementation}.
 
To reconstruct the graph, we reconstruct the traversal and so need
access to two types of data before every step - (1) Is this step
recursive or backtracking (2) If it is a recursive step, then which
neighbour do we recurse to.

To determine the neighbor of the current vertex we need to move to in
a recursive call we need to specify one out $d-1$ possibilities
(except in the first step which has $d$ possibilities). This can be
specified by a tuple of $(d_1,\ldots,d_{\abs{E(\mathcal{H})}}) \in [d]
\times [d-1]^{\abs{E(\mathcal{H})}-1}$ indicating the neighbor. For a
backtracking step, we just pop the stack and thus don't need any
additional data.

We use two ways to figure out whether a step is recursive or
backtracking. The direct way is to just record the sequence in a
binary string of length $2|E\mathcal{H}|$.  A neighbour $u$ of $v$ is
called \textit{recursive} if the edge $(v,u)$ is visited by a
recursive call from $v$. A simple observation about backtracking
sequences is that -- It starts when we encounter a vertex that has
already been visited or we reach a degree one vertex and ends when we
see a visited vertex that has unvisited recursive
neighbors. Therefore, we store a string $\sigma \in
[d]\times[d-1]^{\abs{V(\mathcal{H})}-1}$ in which $\sigma_i$ denotes
the number of recursive neighbors of the $i^{th}$ visited vertex. To
summarize,
\\~\\
\noindent  $\textbf{GraphEnc}(\cH)$: 
\begin{enumerate}[topsep=0.3ex, itemsep=0.4ex, label={({\alph*})}]
  \item Starting vertex $v_1 \in V(G)$
  \item A sequence of degrees $(d_1,\ldots,d_{\abs{E(\mathcal{H})}}) \in [d] \times [d-1]^{\abs{E(\mathcal{H})}-1}$
  \item Either $\sigma \in \set{\textup{R},\textup{B}}^{2\abs{E(\mathcal{H})}}$ (\textbf{Encoding I})
        or $\sigma \in [d]\times[d-1]^{\abs{V(\mathcal{H})}-1}$ (\textbf{Encoding II})
\end{enumerate}
\vspace{0.5cm}
\begin{algorithm_simple}{Unpacking Algorithm for \textbf{GraphEnc}}{$\textup{\textbf{GraphEnc}}(\mathcal{H})$}{$\mathcal{H}$}\label{algo:unpack_hike}
    \begin{itemize}[topsep=0.3ex,itemsep=0.4ex,parsep=0.7ex,label={$\cdot$},leftmargin = 0.4cm]
        \item Initialize DFS stack  $S$ with $v_1$
        \item Initialize $\mathcal{H} = (\set{v_1},\emptyset)$
        \item Initialize $n,r,t = 1$  \qquad \qquad // count visited vertices, recursive steps and total steps 
        \item Initialize $ord(v_1) = 1$
        \item While $S \ne \emptyset$: \begin{itemize}[label={$\cdot$},parsep=0.7ex]
              \item Let $v$ be the top vertex on the stack $S$
              \item $\textup{step} = \textsf{StepType}(v,t)$
              \item If $\textup{step} = \textup{R}$ (recursive): \begin{itemize}[label={$\cdot$}]
                      \item Assign $v_{\text{next}}$ to be $d_r^{th}$ neighbor of $v$ and increment $r$
                      \item Add edge $\set{v,v_{\text{next}}}$ to $\mathcal{H}$
                      \item If $v_{\text{next}}$ is  unvisited : \begin{itemize}[label={$\cdot$}]
                          \item Add vertex $v_{\text{next}}$ to $\mathcal{H}$
                          \item $n \leftarrow n+1$
                          \item $\mathrm{ord}(v_{\text{next}}) \leftarrow n $
                          \item $\textup{push}(v_{\text{next}},S)$
                      \end{itemize}
                      \item Else if $v_{\text{next}}$ is visited, increment $t$ \qquad \qquad // Next step is backtracking
                     \end{itemize}
             \item If $\textup{step} = \textup{B}$ (backtracking): \begin{itemize}[label={$\cdot$}]
                     \item pop(S)           
             \item $t \leftarrow t+1$
              \end{itemize}
        \end{itemize}
        \item return $\mathcal{H}$
    \end{itemize}
\end{algorithm_simple}

\begin{algorithm_simple}{\textsf{StepType}}{$(v,t)$}{$(\text{Type})$}\label{algo:step_type}
    Note - {\small The subroutine to detect the type of step depends on the encoding string $\sigma$.} 
    \begin{itemize}[topsep=0.3ex, itemsep=0.4ex,parsep=0.7ex,label={$\cdot$},leftmargin = 0.4cm]
    \item If $\sigma$ is from \textbf{Encoding I}, return $\sigma_t$
    \item Else, let  $j = \mathrm{ord}(v)$ 
    \begin{itemize}[label={$\cdot$}]
    \item If $\sigma_j >0$ \qquad \qquad//{\small Check if there are any remaining recursive neighbours }
    \begin{itemize}[label={$\cdot$}] 
        \item Decrement $\sigma_j \leftarrow \sigma_j-1$ 
        \item return $R$
        \end{itemize}
     \item Else, return $B$ 
     \end{itemize}
    \end{itemize}
\end{algorithm_simple}

\subsubsection{Counting the encodings}

For the first kind of encoding of type, we have $2^{2k}$ strings of
length $2k$ over $\{R,B\}$. The second encoding might seem wasteful in
general but it is much better when the graph has special structure
that our hike graph will satisfy. We first note that for any vertex
$v$, the number of recursive neighbours
$ \sigma_v \leq \deg_{\cH}(v)-1$ (or $\leq \deg_{\cH}(v)$ if $v =
v_0$).

\begin{definition}[Excess]
The \emph{excess} of $\mathcal{H}$ is defined as $exc(\mathcal{H}) \coloneqq |E(\mathcal{H})| - |V(\mathcal{H})|$.
\end{definition}

\begin{definition}[Excess Set]
We define a vertex to be an \textit{excess vertex} in $\mathcal{H}$ if $deg_{\cH}(v)>2$ and we define the \textit{excess set} to be the set consisting of such vertices i.e  \begin{align*}
  \textup{excSet}(\mathcal{H}) \coloneqq \abs{\set{v \in V(\mathcal{H}) \mid \textup{deg}(v) > 2}}.
  \end{align*}    
\end{definition}

 \begin{lemma}\label{lemma:graphcount}
    Let $G$ be a fixed $d$-regular graph on $n$ vertices. The total number of connected subgraphs $\mathcal{H}$ of $G$ having at most $\leq k$ edges is at most
    \[2 n \cdot d \cdot (d-1)^{k-1} \cdot 2^{2k}. \] 
    Moreover, if $\cH$ is constrained to have at most two vertices of degree one\footnote{We will see later that hike graphs satisfy this strange property } and $exc(\cH) \leq \delta k$, the count is at most 
    \[  2nk^3 \cdot d \cdot (d-1)^{k-1} \cdot 2^{H_2\inparen{\frac{\delta}{1-\delta}}k} \cdot d^{\delta k}. \]
 \end{lemma}
 \begin{proof}
    We first fix the number of edges as $m$ and we will then sum up the expression for $m \leq k$.
    Algorithm~\ref{algo:unpack_hike} unambiguously recovers the graph and therefore the number of possible graphs can be counted by counting the number of possible inputs. The number of degree sequences and start vertices are $n\cdot d (d-1)^{m-1}$. The number of $\sigma$-strings of encoding I are $2^{2m}$. Therefore for a given $m$, we have $n d \cdot (d-1)^{m-1} \cdot 2^{2m}$ and summing this gives the first claim.
    
    In the second case, the key idea is that for every vertex (except the start) of degree $2$, $\sigma_v$ must be $1$. Since $\abs{\textup{excSet}(\mathcal{H})} \leq \delta m$, almost all of the string $\sigma$ is filled by $1$.

We first pick the number of vertices, say $t$. There are at most $m$ choices for this. Then, we let the number of excess vertices be $j$. Summing over all possible $j$, the number of $\sigma$-strings of length $t$ is $ \leq t^2 \sum_{j=0}^{\delta m} \binom{t}{j}d^{j} \le t^2 d^{\delta m} \sum_{j=0}^{\delta m} \binom{t}{j} \leq t^2 d^{\delta m} 2^{H_2\left(\frac{\delta}{1-\delta} \right)t}$.
    
    Here the first term counts the ways or having or up to two vertices of degree 1, the second counts the ways to choose the excess vertices and the third counts the number of their recursive neighbours. In the last inequality we used that $t = m - exc(\cH) \geq (1-\delta)m $. 
    
    The complete expression for the number of graphs would then be 
    \[ \sum_{m\leq k} \left( nd(d-1)^{m-1} \sum_{t = (1-\delta)m}^m t^2 d^{\delta m} 2^{H_2\inparen{\frac{\delta}{1-\delta}}t} \right) ~ \le ~ 2nk^3 \cdot d \cdot (d-1)^{k-1} \cdot 2^{H_2\inparen{\frac{\delta}{1-\delta}}k} \cdot d^{\delta k}. \]
 \end{proof}

\subsection{Bounding Singleton Free Hikes}

Following \cite{MOP20}, we make the following useful definitions,

\begin{definition}[Singleton-free hikes]
  A \textit{$k$-hike} $W$ is a closed walk of $2k$-steps\footnote{That
  is sequence of $(v_0,\cdots, v_{2k-1})$ such that $(v_i,v_{i+1}) \in
  E(G)$ and $v_0 = v_{2k-1}$ } in $G$ in which every step except
  possibly the $(k+1)^{st}$ is non-backtracking. A hike
  is \emph{singleton-free} if no edge is traversed exactly once.
\end{definition}

\begin{definition}[Bicycle free radius \cite{MOP20}]
  A graph $G$ is said to have a bicycle-free radius at radius $r$ if
  the subgraph $\cH$ of distance-$r$ neighborhood of every vertex has
  $exc(\cH) \leq 0$.
\end{definition}
 
We will work with singleton-free hikes in this section. A
singleton-free $k$-hike on $G$ defines a subgraph $\cH$ such that
there at most two vertices of degree $1$ (the start vertex and the
middle vertex) and the number of edges is at most $k$ as every edge is
traversed at least twice. The goal now is to count the possible number
of singleton-free $k$-hikes that yield a fixed subgraph $\cH$. Having
access to $\mathcal{H}$, we will need to encode the hike in a way
similar to the encoding of stale stretches in~\cite{MOP20}.

\noindent \textbf{HikeEnc}: \begin{enumerate}[topsep=0.3ex, itemsep=0.4ex, label={({\alph*})}]
  \item $(v_1,\dots,v_s)  \in V(\mathcal{H})^s$,where $s =\lceil 2k/r \rceil$ and $r$ is the bicycle free radius of $\cH$.
  \item $(c_1,\dots,c_s) \in \{0,\pm 1, \cdots , \pm \lfloor r/2 \rfloor \}^s$. Here, $c_i$ denotes the number of times the unique cycle (in the neighborhood of $v_i$) is to be traversed and the sign indicates the orientation. Since each stretch is of length $r$ and each cycle of length at least $2$ we can traverse a cycle at most $\lfloor r/2 \rfloor$ times.
\end{enumerate}

\begin{claim}\label{claim:hike_enc_in_hike_graph_total}
  For any graph $\mathcal{H}$ that is bicycle free at radius $r$, the number of simple singleton-free $k$-hikes that have $\cH$
  as their hike graph is at most $\left(\abs{ rV(\mathcal{H})}\right)^{\lceil 2k/r \rceil}$.
\end{claim}
\begin{proof}
  Follows from the possible values the encoding \textbf{HikeEnc} can take.
\end{proof}

We use a generalization of the bound of Alon \etal \cite{AHL02} on the
excess number (originally involving the girth), extended to
bicycle-free radius in~\cite{MOP20}.

\begin{theorem}{\cite[Theorem 2.13]{MOP20}}
  Let $\mathcal{H}$ be a bicycle free graph of radius $r \ge 10 \ln(\abs{V(\mathcal{H})})$. Then
  \begin{align*}
    \textup{exc}(\mathcal{H}) \le \frac{\ln(e\abs{V(\mathcal{H})})}{r} \cdot \abs{V(\mathcal{H})}.
  \end{align*}
\end{theorem}

\begin{corollary}\label{cor:excessdegree}
   Let $G$ be a $d$ regular graph on $n$ vertices bicycle free at radius $r$.
   Let $\mathcal{H}$ be a subgraph with at most two vertices of degree one on $n_0$ vertices where $n_0 = e^{\delta r - 1} $ for some $\delta \leq 1/10$. Then, 
   \[  \textup{excSet}(\mathcal{H}) \le 2\delta n_0 +2. \]
\end{corollary}
\begin{lemma}\label{lemma:total_hike_graphs_2}
   Let $G$ be a $d$ regular graph, with $d \ge 3$, on $n$ vertices bicycle free at radius $r$. Then, the total number of singleton free $(k-1)$-hikes on $G$ is at most 
 \begin{align*}
        \left(2^{\gamma_1}\sqrt{d-1}\right)^{2k} \text{ where } \gamma_1 = 1 + \frac{\log (nrk)}{2k} + \frac{\log(rk)}{r}.
  \end{align*}
   
   If we assume that $3 \le k \le e^{\delta r} $, then it is at most
     \begin{align*}
        \left(2^{\gamma_2}\sqrt{d-1}\right)^{2k} \text{ where } \gamma_2 = \frac{\log (16 nk^3rd)}{2k} + \frac{\log(rk)}{r} + H_2(5\delta)/2 +  \delta \log d.  
  \end{align*}
  \end{lemma}

\begin{proof}
  Any singleton-free $(k-1)$-hike defines a connected graph $\cH$ with at most $k-1$ edges and therefore at most $k-1$ vertices. If there is no backtracking step then all vertices except the start have degree at least two. Else, the end point of one of the backtracking step may have degree $1$. Thus there are at most 2 vertices of degree one. When $k$ is unbounded, we use the bound from the first encoding i.e. \cref{lemma:graphcount} and combine it with the number of possible hikes on this from \cref{claim:hike_enc_in_hike_graph_total} to get 
   \begin{align*}
     &\leq 2 n \cdot d \cdot (d-1)^{k-2} \cdot 2^{2(k-1)} (r(k-1))^{\frac{2(k-1)}{r}+1}\\
     &\leq  (nrk)\cdot (d-1)^{k} \cdot 2^{2k}  (rk)^{\frac{2k}{r}}\\
     &\leq  \left(2\cdot 2^{ \log (nrk)/2k} 2^{ \frac{\log(rk)}{r} }  \right)^{2k} (d-1)^{k} \\
     &\leq  \left(2^{\gamma_1}\sqrt{d-1}\right)^{2k}.
 \end{align*}
  
  The assumption on $k$ lets us use~\cref{cor:excessdegree} which when combined with \cref{lemma:graphcount} gives us the bound on the number of such graphs as $  4 nk^2 d \cdot (d-1)^{k-1} \cdot {k \choose 2\delta k+1} \cdot d^{2\delta k+1} $. Combining with the number of possible hikes on this from \cref{claim:hike_enc_in_hike_graph_total}, we get the total number of singleton-free $k$-hikes bounded by 
 \begin{align*}
     &\leq 4 n(k-1)^2 \cdot d \cdot (d-1)^{k-2} \cdot {k-1 \choose 2\delta (k-1)+2} \cdot d^{2\delta (k-1) +2} (r(k-1))^{\frac{2k-2}{r}+1}\\
     &\leq  ( 16 nk^3rd) (d-1)^{k} \cdot 2^{H_2(5\delta)k} \cdot d^{2\delta k} (rk)^{\frac{2k}{r}}\\
     &\leq  \left(2^{ \log (16 nk^3rd))/2k} d^{\delta}2^{ \frac{\log(rk)}{r} } 2^{H_2(5\delta)/2} \right)^{2k} (d-1)^{k} \\
     &\leq  \left(2^{\gamma_2}\sqrt{d-1}\right)^{2k}.
 \end{align*}
\end{proof}

\section{Instantiation of The First Two Main Results}
In this section, we will use the bound on singleton-free hikes
obtained in the last section to bound the eigenvalue of the lifted
graph. We first handle non-singleton free hikes and show that they can
be easily bounded by the $\varepsilon$-biased property of the
distribution of the signings. We then formalize the construction by
instantiating it using an expander from MOP having large bicycle-free
radius and then bring the bounds together.

\subsection{A Simple Generalization of The Trace Power Method in MOP}

We now show that the problem of bounding the spectral radius of the
signed non-backtracking operator reduces to counting singleton-free
hikes. This reduction is a straightforward generalization of the
argument~\cite[Prop. 3.3]{MOP20} for $\Z_2$ to any abelian group.

Let $B_s(\chi)$ (as defined in \cref{def:backtrack}) be the signed
non-backtracking operator with respect to a signing and a non-trivial
character $\chi$ and $\rho(B_s)$ denote its spectral radius.  The goal
is to bound the largest eigenvalue of $B_s(\chi)$. The trace method is
the name for utilizing the following inequality,
$$
\tr((B^*)^{k}B^{k}) = \norm{B^k}_F = \sum_i\abs{\lambda_i^k}^{2}  \geq \rho(B)^{2k}.
$$
The signing $s$ is drawn from some distribution $\cD$ and we wish to
show via the probabilistic method that there exists a signing in $\cD$
for which $\rho(B_s(\chi))$ is small for any set of $(l-1)$
non-trivial characters $\chi$. We will use a first-order Markov
argument and therefore wish to bound $\E_{s \sim \cD }\tr(B_s^{k}
(B_s^*)^{k})$. Writing it out we get,
\begin{align*}\label{eqn:split_trace}
    T_\chi(s) = \tr((B_s^*)^{k}B_s^{k} ) &= \sum_{e \in E^d} \left((B_s^*)^kB_s^ke\right)_{e}\\
    &=\sum_{ (e_0,\cdots, e_{2k})} B(e_0,e_1)\cdots B(e_{k-1},e_k) B^*(e_{k},e_{k+1})\cdots B^*(e_{2k-1},e_{2k}) \\
    &=\sum_{ (e_0,\cdots, e_{2k}) } \chi(s(e_1))\cdots \chi(s(e_{k})) \chi^*(s(e_{k}))\cdots \chi^*(s(e_{2k-1})) \\
    &=\sum_{ (e_0,\cdots, e_{2k}) } \chi(s(e_1))\cdots \chi(s(e_{k-1})) \chi^*(s(e_{k+1})) \cdots \chi^*(s(e_{2k-1})).
\end{align*}

Notice that $e_0, e_k$ don't appear in the term and so we define
$\cH_{k-1}$ as the multiset of all tuples $(e_1, \ldots,
e_{k-1},e_{k+1}, \ldots ,e_{2k-1})$ appearing in the support of this
summation. We denote each term in the summation above by $\chi_w(s)$
where $w \in \cH_{k-1}$. It follows directly from the definition that
each $w \in \cH_{k-1}$ defines a $(k-1)$-hike. Also observe that, any
tuple appears at most $(d-1)^2$ times as given a tuple $w$, we have at
most $(d-1)$ choices for each $e_0, e_k$. Let $\cH^s_{k-1}$ denote the
singleton-free hikes in $\cH_{k-1}$. We can split $T_\chi(s) = T_1(s)
+ T_2(s)$ where
\[ T_1(s) = \sum_{ w  \in \cH^s_{k-1} } \chi_{w}(s),\;\; T_2(s) = \sum_{ w  \not\in \cH^s_{k-1} } \chi_{w}(s) .\]

 We now define $\varepsilon$-biased distributions that will be the key pseudorandomness tool.

\begin{definition}[Bias]
  Given a distribution $\mathcal{D}$ on a group $H$ and a character
  $\chi$, we can define the bias of $\mathcal{D}$ with respect to
  $\chi$ as $bias_\chi(\mathcal{D})
  := \abs{\E_{h \sim \mathcal{D}} \chi(h)}$ and the bias of
  $\mathcal{D}$ as $\bias(\mathcal{D}) = \max_{\chi}
  bias_\chi(\mathcal{D})$, where the maximization is over non-trivial
  characters.
\end{definition}

\begin{lemma}\label{lem:biased_sum}
  Let $\cD \subseteq H^{E(G)}$ be an $\nu$-biased distribution and let
  $w \not\in \cH^s_{k-1} $ be a singleton-hike i.e. there is an edge
  that is travelled exactly once. Then, $|\E_{s\sim \cD} \chi_{w}(s)
  | \leq \nu$.
\end{lemma}

\begin{proof}
   Let the set of distinct edges in $w$ be $\{e_1, \cdots, e_r\}$ and
   let edge $e_i$ be travelled $t_i$ times where $t_i$ takes the sign
   into account.\footnote{ Let $e_i$ appear $f_1$ times in the first
   $k-1$ steps and $b_1$ times in the next $(k-1)$ steps.  Similarly
   let $e_i^T$ which is the reverse direction of $e$ appear $f_2$
   times in the first $k-1$ steps and $b_2$ times in the next $(k-1)$
   steps.  Then, $t_i = f_1+b_2 - f_2-b_1$.}Let $e_j$ be the edge
   traversed exactly once. Then, $t_j = \pm 1$. Now, we can rewrite
   $\chi_w(s) = \prod_{i=1}^r \chi(s(e_i))^{t_i} $ and it can be
   extended to a character on $H^{E(G)}$. Since $t_j =\pm1$, this
   character is non-trivial and the claim follows from the
   $\nu$-biased property.
\end{proof}

\begin{lemma}[Analog of Corr. 3.11 in \cite{MOP20}]\label{lem:spectrum}
  Let $G$ be a $d$-regular graph on $n$-vertices, $\varepsilon < 1$ be a fixed constant, $\ell$ be a parameter,
  $H \subseteq Sym(\ell)$ be an abelian group and $\cD \subseteq H^m$ be an $\nu$-biased distribution such that
  $\nu \leq (nld^2)^{-1}. \left(\frac{\varepsilon}{d}\right)^{2k}$. 

  Assume that the number of singleton-free $(k-1)$-hikes is bounded by $(2^\gamma\sqrt{d-1})^{2k}$.
  Then for any non-trivial character $\chi$ of $H^m$, we have that except with probability at most $1/\ell$ over $\cD$,
  $\rho( B(\chi)) \leq 2^{\gamma'}\sqrt{d-1} + \epsilon $ where $\gamma' = \gamma + \frac{\log(\ell d^2)}{2k}$.
\end{lemma}

\begin{proof}
  By the decomposition above, we have $T(s) = T_1(s) + T_2(s)$.
  As each term in the expression is of the form $\chi(h)$ and as remarked earlier,
  all the characters are roots of unity so $\abs{\chi(s(e))} = 1$.
  Thus, $\abs{T_1(s)} \leq  \abs{\pi^{-1}\left(\cH^*_{k-1}\right) } \leq (d-1)^2\abs{\cH^*_{k-1}}$ 
\begin{align*}
  \mu := \abs{\E_{s \sim \cD} T}  &= \abs{\E T_1  + \E T_2} \\
                              &\leq \abs{\E T_1}  + \abs{\E T_2} \\
                              & \leq \abs{\cH^s_{k-1}}  + \sum_{w \not\in \cH^s_{k-1}} \abs{\E_{s \sim \cD}\chi_w(s)}\\
                              & \leq d^2(2^\gamma\sqrt{d-1})^{2k}  + \nu \abs{\cH_{k-1}}\\
                               & \leq d^2(2^\gamma\sqrt{d-1})^{2k}  + \nu nd^{2k+2}.
\end{align*}
Here we have used the observation that $\abs{\cH^s_{k-1}} \leq
(d-1)^2\{|\text{Singleton-free $(k-1)$-hikes}|\}$ and \cref{lem:biased_sum}. The bound on $\abs{\cH_{k-1}}$ is trivial as we have $nd$ choices for the starting edge and a walk of length of $2k+1$. Since $T$ is a
non-negative random variable, we apply Markov to conclude that
$T \leq \mu \ell$ with probability at most $1/\ell$.
\begin{align*}
\rho(B_s(\chi))\le T^{1/2k} < (\mu \ell)^{1/2k} &\leq  \left(  d^2\ell\left( 2^\gamma\sqrt{d-1} \right)^{2k} + \nu \ell n d^{2k+2}\right)^{1/2k}\\
  &\leq   (d^2\ell)^{1/2k}2^\gamma\sqrt{d-1} +  \left(\nu \ell n d^{2k+2} \right)^{1/2k}\\
  &\leq  2^{\gamma'}\sqrt{d-1} +  (\nu \ell n d^{2})^{1/2k}d\\
  &\leq  2^{\gamma'}\sqrt{d-1}  +  \frac{\varepsilon}{d}d\\
  &\leq  2^{\gamma'}\sqrt{d-1} + \varepsilon.
\end{align*}
\end{proof}

\subsection{The Instantiation}

Before we instantiate the explicit construction of abelian lifted
expanders leading to~\cref{theo:main1}, we will need two tools.
The first one is an explicit construction of expander graphs to be
used as base graphs in the lifting operation. Since we need this
technical condition of bicycle-freeness, we use the construction
in~\cite{MOP20}.

\begin{theorem}\cite[Theorem 1.1]{MOP20}\label{thm:mopgraph}
For any given constants $d\ge 3,\varepsilon >0$, one can construct in
deterministic polynomial time, an infinite family of graphs $\{G_n\}$
with $\lambda(G_n) \leq 2\sqrt{d-1}+\varepsilon$ and $G_n$ is
\begin{itemize}
    \item $n \le |V(G_n)| \le 2n$ ,
    \item $G_n$ is bicycle-free at radius $c \log_{d-1}( |V(G_n)|)$,
    \item $\lambda_2(B_G) \leq \sqrt{d-1} + \varepsilon$.
\end{itemize}
\end{theorem}

The second tool is a $\nu$-biased distribution for abelian groups
(having a sample space depending polynomial on $1/\nu$). In
particular, we use a recent construction by Jalan and Moshkovitz.

\begin{theorem}\cite{JM21}\label{thm:biasedset}
  Given the generating elements of a finite abelian group $H$ and an integer $m \geq 1$ and $\nu >0$, there is a deterministic polynomial
  time algorithm that constructs subset $S \subseteq H^m$ with size $\;O\left(\frac{m \log(H)^{O(1)}}{\nu^{2+ o(1)}}\right)$ such that the
  uniform distribution over $S$ is $\nu$-biased.
\end{theorem}

We are now ready to prove our first main result.

\MainI*

\begin{proof}
  Construct $G_0$ on $n \le n' \le 2n$ vertices for given $(d,\epsilon)$ using \cref{thm:mopgraph} which has $r \geq c\log_{d-1} n'$. 

\begin{itemize}[leftmargin=0.3cm, label={$\cdot$}]
    \item \textbf{Regime 1} - Here shorter walks will suffice and we will use the bound on $\gamma_2$ from \cref{lemma:total_hike_graphs_2}. To get Near-Ramanujan, we need $\gamma' = \gamma_2 + \frac{\log (d^2\ell)}{2k} = \gamma_2' + \frac{\log (\ell)}{2k} $ to be vanishing with $\varepsilon$. Observe that when $k = \omega(\log n)$, $\gamma_2$ is bounded by $ o(1) + \left(2\sqrt{\delta} + \delta\log d\right) $. We pick $\delta$ small enough and assume that $n' \ge N(\varepsilon,d)$ such that $\gamma_2' \leq \frac{2\varepsilon}{\sqrt{d-1}}$. In the bounded $k$ regime we can pick $k < e^{\delta r}$. Since, $\frac{\log (\ell)}{2k}$ must also be vanishing in $\varepsilon$, this forces $\log(\ell) \leq \varepsilon k \le \varepsilon e^{\delta r}$. This explains the bound on $\ell$.
    
    \item \textbf{Regime 2} - Here $\ell$ is larger and so we pick $k = \log \ell$. Now, we need to use $\gamma_1$ which we recall is $1 + \frac{\log k}{r} + o(1)$. Thus, $\gamma' = (\gamma_1 + \frac{\log{d^2}}{k}) + \frac{\log{\ell}}{k} \leq 3/2 + \frac{\log k}{r}$. Since, $r = c \log_{d-1}(n')$, to get non-trivial expansion $k \le n^{c/2}$ which explains the bound on the exponent $\delta$.   
\end{itemize}

The precise parameters are as follows
\[
\begin{array}{|c|c|c|c|c|c|}
     \hline
     \textbf{Regime} & \delta &  k & \nu & \gamma' \\
     \hline
     1& O\left(\frac{\varepsilon^2}{d}\right)& \frac{10\sqrt{d-1}}{\varepsilon}\max(\log \ell, \log n)& (nld^2)^{-1} \left(\frac{\varepsilon}{3d}\right)^{2k} = (n\ell)^{c_{d,\varepsilon}} & \frac{2\varepsilon}{3\sqrt{(d-1)}}\\
     \hline
     2&  \le c/2 & \log \ell = n^\delta & (n \ell d^2)^{-1} \left(\frac{1}{3d}\right)^{2k} = (n\ell)^{c_{d}}  & 2 + \frac{\delta}{c}\log (d-1)  \\
     \hline
\end{array}
 \]

Construct a $\nu$-biased distribution $\cD$ using
\cref{thm:biasedset}. These two constructions take $\poly(n,\ell)$
time.

From \cref{cor:decomposespec}, we have to analyze $B(\chi)$
for $\ell-1$ non-trivial characters $\chi$ that appear in this
decomposition. The largest eigenvalue is clearly given by $B(1)$ which
is $d-1$. For the second largest, $\lambda_2(B(1)) \leq \sqrt{d-1} + \varepsilon$ by
the property of the base graph $G$ obtained by
\cref{thm:mopgraph}. Since we have the bicycle-free property,
we can use~\cref{lem:spectrum} to conclude that for any
non-trivial characters we have except with probability at most $1/\ell$
\begin{itemize}[leftmargin=0.3cm, label={$\cdot$}]
    \item \textbf{Regime 1} - $\rho(B(\chi)) \leq 2^{\gamma'}\sqrt{d-1} + \epsilon/3 \leq \sqrt{d-1} + \varepsilon$.
    \item \textbf{Regime 2} - $\rho(B(\chi)) \leq 2^{\gamma'}\sqrt{d-1} + 1 \leq 2\cdot2^{2} d^{\delta/c} \sqrt{d-1}  ~\le~ \varepsilon d $ when $\; d \ge \left(\frac{8}{\varepsilon}\right)^{\frac{2c}{c-2\delta}}$.
\end{itemize}

Using the \cref{fact:character_decomp}, we assume that the
decomposition has exactly one trivial character (say, $\chi_1$) and
$(l-1)$ non-trivial characters. Then for the trivial character
$\rho(B_{G_0(s)}) = \rho(B(\chi_1)) = d-1$ and thus, $ \rho_2(B)
= \max \left\{ \lambda(G_0), \max_{i=2}^\ell \rho(B(\chi_i)) \right\}$.

Since the bound holds for any non-trivial $\chi$ except with
probability $1/\ell$ we take a union bound over these $\ell-1$
characters we get that there is a labelling $s \in D$ such that the
bound holds for $\rho (B(\chi_i))$ and thus for
$\lambda(B_{G_0(s)})$. By \cref{fact:ihara}, we get that
$\lambda(G) \leq 2\rho_2(B_G) $ which satisfies the bounds we need.

We can brute force through each $s \in \supp(\cD)$ to find an $s$ such
that the lifted graph $G = G_0(s)$ has the required spectral
gap. Checking this is a simple linear algebraic task and can be done
in time cubic in $n\ell$. Therefore, the total time taken is
$\poly(n,\ell)$.
\end{proof}

\section{Derandomizing Exponential Lifts}\label{sec:exact_exponential_lift}

We will now construct explicit expanding graphs where the lift size is
exponential. In this regime, known tools like expander Chernoff
suffice and in fact, one can verify that the results of \cite{ACKM19}
can be directly derandomized by an application of these. However, we
give a simplified (mostly) self-contained proof relying on a key lemma
of Bilu and Linial \cite{BL06} which could be of independent interest
and derandomize it.

\subsection{The Setup and Construction}

In this subsection, we will describe how the signings for the lift are
generated via walks on an expander and utilize an expander Hoeffding
bound which will be used to bound the spectrum. We will assume from
now that the group is $\Z_\ell$\footnote{This can be extended in a
straightforward manner to any abelian subgroup $H \leqslant Sym(\ell)$
by just taking the expander graph on $|H|$ vertices. Since, we only
work with abelian groups having a transitive action, $\abs{H} = \ell$.
} We start with a expander construction from \cite{Alon21},

\begin{theorem}\cite[Thm. 1.3]{Alon21}
  For every degree $d \ge 3$, every $\epsilon > 0$ and all sufficiently large $ m \geq n_0(d,\epsilon) $ where
  $md$ is even, there is an explicit construction of an $(m, d, \lambda)$-graph with $\lambda \leq 2\sqrt{d-1} + \epsilon$.
\end{theorem}

We can fix the degree to be an even constant, say $d'$, and have
$\epsilon = \sqrt{d'-1}$. Then, for every large enough $m$, we have an
expander on $m$ vertices. We use this to get an explicit expander,
$L$, on $\ell$ vertices. To obtain a sequence of lifts i.e. elements
of $\Z_\ell$, we first pick a random vertex, $v_1$, of $L$ which uses
$\log \ell$ bits of randomness. Then, we do a random walk for $dn-1$
steps producing a sequence $(v_1, \cdots, v_{dn-1})$ of vertices of
$L$ which we interpret as elements of $\Z_\ell$. Each step of the
random walk requires $O(\log d')$ bits of randomness as the graph is
$d'$-regular. Therefore, the total amount of randomness is
$O(\log \ell + (dn-1)\log d' )$\footnote{Another way to say this is
that the number of walks of length on $dn$ on $L$ is $|H| \cdot
d'^{dn}$.} The main observation is that an expander random walk
suffices to guarantee that the lifted graph will be an expander. To
formalize this, we first state a Hoeffding type concentration result
for our random variables generated via the Markov chain on an
expander.

\begin{theorem}\cite[Thm. 1.1]{Rao19}\label{thm:expanderchernoff}
  Let $\{Y_i\}_{i=1}^\infty$ be a stationary Markov chain with state space $[N]$, transition matrix $A$,
  stationary probability measure $\pi$, and averaging operator $E_\pi$, so that $Y_1$ is distributed
  according to $\pi$.
  Let $\lambda = |A - E_\pi|_{L_2(\pi) \to L_2(\pi)}$ and let $f_1, \cdots , f_t: [N] \to \R$ so that $\E[f_i(Y_i)] = 0$
  for all $i$ and $|f_i(v)| \leq a_i$ for all $v \in [N]$ and all $i$. Then for $u \geq 0$,
  \begin{align*}
    \Pr\left[\left|\sum_{i=1}^t f_i(Y_i)\right| \geq u \left(\sum_{i=1}^t a_i^2\right)^{1/2}\right] \leq 2\exp\left(\frac{-u^2(1-\lambda)}{64e}\right).
  \end{align*}
\end{theorem}

\begin{corollary}\label{cor:expandercher}
Let $U$ be any subset of edges in the base graph $G$. Then, 
$$\Pr\left[ \left|\sum_{ e \in U } \Re(\chi(s(e)))\right| \geq t \right] \leq 2\exp\left(\frac{-t^2}{128e|U|}\right), $$

$$\Pr\left[ \left|\sum_{ e \in U } \Im(\chi(s(e)))\right| \geq t \right] \leq 2\exp\left(\frac{-t^2}{128e|U|}\right). $$

\end{corollary}
\begin{proof}
Let $Y_{e} = s(e)$ be the random variables associated to each edge
$e$. From the construction described earlier, $\{Y_{u,v}\}$ is a
Markov chain with the transition matrix being the weighted adjacency
matrix of the expander $L$ with second (normalized)eigenvalue bounded
by $3\sqrt{d'-1}/d'$. Thus, $1-\lambda \geq
1- \frac{3}{\sqrt{d'}} \geq 1/2$ for $d' \geq 36$. The stationary
measure $\pi$ is the uniform measure on vertices of $L$ and it is
stationary as the all-ones vector is an eigenvector of the weighted
adjacency matrix with eigenvalue 1. Recall that we picked the first
vertex ($Y_1$) uniformly i.e. from $\pi$. Let $f_{e}
= \Re(\chi(s(e)))$ if $e \in U$ and $0$ otherwise. Analogously $g_{e}
= \Im(\chi(s(e))$ when the edge is in $U$ and $0$ otherwise.

$\E[f_{e}] = \frac{1}{\ell} \sum_{i=0}^{\ell-1} \Re(\omega^i) $ because the
characters are roots of unity and the expectation is over $\pi$ which
is uniform. Since the sum of roots of unity are zero, so is its real
and imaginary part. This holds thus for $g_{e}$ too. Moreover, $a_{e}
= 1$ if $e \in U$ and is $0$ otherwise.
Applying \cref{thm:expanderchernoff} with $u := t/\sqrt{|U|}$ gives the
result.
\end{proof}

\subsection{A Simpler Lifting Proof}

In this section, we give a simpler proof of a weaker result similar to
one in \cite{ACKM19} which says that if the lifts were picked
independently and uniformly at random, then the lifted graph is also
expanding. In place of the random signings, we will use the signings
generated from random walks as described in the earlier section.  The
proof can be seen as building up from the tools of
\cite{BL06} by simplifying their derandomization of the 2-lift and
extending it to $\ell$-lifts.

Let $G$ be a graph, $H$ be an abelian group and $s$ be a signing
$s:E\to H$ that gives a $(H,\ell)$-lift. Let $A$ be its adjacency
matrix. We will use the earlier notation and denote by $A(\chi)$, the
matrix where for every edge $e$, we replace $1$ by $\chi(s(e))$. Let
$A(\chi) = C + iD$ where $C, D$ are real symmetric matrices.  We want
to bound the spectral radius of $A(\chi)$. It is not very hard to see
that $\norm{A} \leq 2 \max\{ \norm{C}, \norm{D}\} $. This can be
observed by letting $v = v_1 + iv_2$ be an eigenvector and $\alpha
= \max\{ \norm{C}, \norm{D}\} $. Then,
\begin{align*}
  v^* Av = \Re(v^* A v) &= (v_1^TCv_1 + v_2^TCv_2 - v_1^TDv_2 + v_2^TDv_1)\\
         &\leq \norm{C}\norm{v_1}^2 + 2\norm{D}\norm{v_1}\norm{v_2}\\
         &\leq \alpha (\norm{v_1} +\norm{v_2})^2\\
         &\leq 2 \alpha \left(\norm{v_1}^2 + \norm{v_2}^2 \right)\\
         &= 2 \alpha \norm{v}^2.
\end{align*}
Therefore we reduce the problem to bounding spectral radius for the
constituent real matrices. We now state a very useful lemma by Bilu
and Linial which is a discretization result and can be seen as a
converse to the expander mixing lemma. It says that bounding the
Rayleigh coefficient on Boolean vectors suffices to bound the real
spectrum up to logarithmic factors.

\begin{theorem}\cite[Lemma 3.3]{BL06}\label{thm:bilu1}
Let $A$ be an $n \times n$ real symmetric matrix such that the $l_1$ norm
of each row in $A$ is at most $d$, and all diagonal entries of $A$ are, in absolute
value, $O(\alpha(\log(d/\alpha) + 1))$. Assume that for any two vectors, $u,v \in \{0,1\}^n$,
with $\supp(u)\cap \supp(v) = \emptyset$:

\[ \frac{|u^tAv|}{\norm{u}\norm{v}} \leq \alpha,\]
then spectral radius of $A$ is $O(\alpha(\log(d/\alpha) + 1))$.
\end{theorem}

Since the graph is $d$-regular, $A(\chi)$ is $d$-sparse and so is $C$
and $D$. The $\ell_1$-norm of any row of $C, D \leq d $ as we have a
sum of $d$ entries of the form $\Re(\omega^j), \Im(\omega^j)$ for some
$j$ and the absolute value of each of these is upper bounded by
$1$. Moreover, the diagonal entries are all zero. Therefore, these
satisfy the norm criteria of the theorem. Now, we need to bound the
(generalized) Rayleigh coefficient.

Let $S,T$ be subsets of the vertices of a $d$-regular graph. Define
$E(S,T) = \{(x,y) \in E\;|\; x \in S, \;y\in T\}$ and let $e(S,T) :=
|E(S,T)|$. Let $u,v \in \{0,1\}^n$ and let $S := \supp(u),
T:= \supp(v)$. Then,
\begin{align}\label{eqn:sumbound}
     \abs{u^TCv} \leq \sum_{u \in E(S,T)} |\Re(\chi(s(e)))| \leq e(S,T),   \\ 
     \abs{u^TDv} \leq \sum_{u \in E(S,T)} |\Im(\chi(s(e)))| \leq  e(S,T).  
\end{align}
Let us now state the expander mixing lemma,
\begin{equation}\label{eqn:eml}
\left|e(S,T) - \frac{d|S||T|}{n} \right| \leq \lambda(G) \sqrt{|S||T|} .
\end{equation}

Now we assume that the signing $s$ was generated from the random walk
as described earlier. This lets us use \cref{cor:expandercher}.

\begin{lemma}\label{lem:bilu2}
  Let $M$ be either $C$ or $D$ where these are the matrices defined above. Pick $\gamma$ such that $\gamma^3 \geq \frac{256e \sqrt{d}}{n} \ln(3\ell)$.
  Then for every pair of vectors $u,v \in \{0,1\}^n$, $|u^TMv| \leq \alpha \norm{u}\norm{v}$
  where $\alpha = (\gamma+1)\lambda$ except with probability $\frac{2}{3\ell}$ over choice of $s$. 
\end{lemma}
\begin{proof}
Since the proofs are identical, we use $M$ as a placeholder which can
be replaced by $C$ or $D$.  Let $S := \supp(u), T:= \supp(v)$ and
define $a := \sqrt{\norm{u}\norm{v}} = \sqrt{|S||T|}$. 
\paragraph{Case 1 - $a \leq \frac{\gamma n\lambda }{d}$.}From  \cref{eqn:sumbound} and \cref{eqn:eml},  we have,
\[ |u^TMv| ~\leq~ e(S,T) ~\leq~ \frac{d}{n}a^2 + \lambda a ~\leq~ (\gamma+1) a\lambda. \] 

\paragraph{Case 2 - $a > \frac{\gamma n\lambda }{d}$.}Using the trivial bound that $a \leq n$ in ~\cref{eqn:eml}, we get, \[e(S,T) ~\leq~ a(\frac{da}{n} + \lambda) ~\leq~ a(d+\lambda) \leq 2ad. \]By \cref{cor:expandercher} we get, 
\[ \Pr\left[ \left|u^TMv \right| \geq (\gamma+1) a\lambda \right] ~\leq~ 2\exp\left(\frac{-((\gamma+1) a\lambda)^2}{128e\cdot e(S,T) }\right). \]
We can upper bound this as,
\begin{align*}
\frac{( (\gamma+1) a\lambda)^2}{128e\cdot e(S,T)} &\geq \frac{((\gamma +1) a\lambda)^2}{128e (2ad) }\\
  &\geq  \frac{ (\gamma+1)^2 a}{256e } && \text{($\lambda^2 > d$)} \\
  &\geq   \frac{\gamma(\gamma+1)^2 n \lambda}{256e d} && \text{(By case assumption on $a$)}\\
  &\geq   \frac{\gamma^3 n}{2\cdot 256e \sqrt{d}} \\
  &\geq \ln(3l). && \text{(By assumption on $\gamma$)}
\end{align*}
\end{proof}

\begin{theorem}[Exactly Exponential Lifts]
 For any positive integers $n,\ell$ and every constant degree $d$, there exists a deterministic $\poly( \exp(n),\ell)$ time algorithm that constructs a $d$-regular graph $G$ on $n\ell$ vertices such that  
  \begin{itemize}
    \item $G$ is quasi-cyclic with lift size $\ell$, and
    \item  If $\ell \leq \exp\left( cnd^{-1/2} \right)$, then $\lambda(G) \le O(\sqrt{d} \log d)$.
    \item  If $\ell = \exp\left( cnd^\delta \right)$ for $\delta \in [-1/2, 1)$, then $\lambda(G) \le O\left(d^{\frac{2+\delta}{3}} \log d\right)$.
  \end{itemize}
In particular, this yields an explicit polynomial time construction of a lift in the regime when $\ell = \exp(O(n))$.
\end{theorem}

\begin{proof}
We construct a $d$-regular graph $G_0$ using \cite{Alon21} on $n$
vertices such that $\lambda_2(G) \leq 2 \sqrt{d}$. We generate a set
of signings as described above using a $d'$-regular expander on $\ell$
vertices. This takes time $\ell \exp(nd\ln (d'))$ and we can fix $d' =
36$. or each signing we compute the eigenvalue of the adjacency matrix
of the lifted graph and pick the one with the smallest
$\lambda_2$. The existence of a good signing is guaranteed as follows.

\cref{lem:bilu2} gives a bound of $\alpha = (\gamma+1)\lambda$ on the Rayleigh quotient of $C, D$ holds except with probability $\frac{2}{3\ell}$ over the signings.
\cref{thm:bilu1} then implies that
$$\lambda_{\max}A(\chi) \leq 2\max\{\norm{C}, \norm{D}\} \leq  2\alpha \log(d/\alpha) \leq O( \alpha \log d ). $$

Here, $\gamma^3 = O( \frac{\ln \ell\sqrt{d} }{n} ) = O(d^{1/2 + \delta})
$.  Note that if $\delta <-1/2$, then $O( (\gamma +1)\lambda ) =
O(\lambda)$ Thus, $\alpha = O(\sqrt{d} d^{1/6 + \delta/3} + \lambda )
= O\left(\max \left\{
d^{\frac{2+\delta}{3}}, \sqrt{d} \right\} \right)$.

To finish the proof we need to take a union bound over each of the
$\ell-1$ non-trivial characters and bound the spectrum of $A(\chi)$ as
above. Thus, we have that the probability that there exists a good
signing is at least $1 - \ell \left(\frac{2}{3\ell}\right) > 0$.
\end{proof}

\section{Explicit Quantum and Classical Codes}
We now briefly recall the construction of quantum LDPC codes as
in \cite{PK20} and show how our results derandomize it.  The
construction is as follows. Let $G$ be a $d$-regular graph (on $n\ell$
vertices) such that $G$ is a $(\Z_\ell, \ell)$-lift of a graph on
$n$-vertices. Let $\mathcal{C}_0 \subseteq \F_2^d$ be a binary linear
code (of block length $d$). Let $B$ denote the bipartite graph of the
Tanner code $\textup{T}(G,\mathcal{C}_0)$ and let $F$ denote the cycle
graph on $\ell$ vertices. They define the lifted product
$\textup{LP}(B,F)$ of $B$ and $F$ which is a variation of the usual
tensor product and is also equivalent to the twisted product
in \cite{hho20}.  The main result of \cite{PK20} is the following.

\begin{theorem}[\cite{PK20}]\label{theo:qlp_codes_pk}
   Let $G$ be $(\Z_\ell, \ell)$-lift of a $d$ regular graph on $n$-vertices with $\lambda_2(G) \leq \varepsilon \cdot d$.
   Let $\mathcal{C}_0 \subseteq \mathbb{F}_2^d$ and its dual attain the Gilbert--Varshamov bound.
   If $\varepsilon > 0$ is sufficiently small and $d$ is a sufficiently large constant, then
   
   \begin{itemize}
       \item $\textup{T}(G,\mathcal{C}_0)$ is a good quasi-cyclic LDPC code of blocklength $\Theta(n\ell)$ and circulant size $\Theta(\ell)$.
       \item The quantum \textit{lifted product} code $\textup{LP}(B, F)$ is LDPC and has distance $\Theta_{\epsilon,d}(\ell)$ and dimension $\Theta(n)$.
   \end{itemize}
\end{theorem}

To achieve these, \cite{PK20} picks a $d$-regular expander on $n$
vertices and creates a random $\ell$-lift i.e. where each signing is
chosen uniformly at random from $\Z_\ell$. The final graph is
expanding with high probability from the results of \cite{ACKM19}
(~\cref{theo:agarwal_et_al} ). The distance achieves the almost-linear
bound only when the lift is large and thus lifts of exponential size
are preferred. By the upper bound in ~\cref{theo:agarwal_et_al},
better than exponential size lifts break expansion for abelian groups.

For this application, the constant degree regime is important for two
reasons. The locality of the code is essentially $d$ and thus it has
to be constant for it to be LDPC. Moreover, the code
$\mathcal{C}_0 \subseteq \mathbb{F}_2^d$ can be easily constructed via
brute-force search since $d$ is constant.

While the corollary follows in a straightforward manner from our main
results, we show the computations for completeness.

\begin{corollary}[\cite{PK20}, \cref{theo:main1}, \cref{theo:main2}]
   We have explicit polynomial time construction of each of the following,
   \begin{enumerate}
       \item Good quasi-cyclic LDPC code of block length $N$ and
  any circulant size up to $N/\polylog(N)$ or $\Theta(N/\log(N))$.
       \item Quantum LDPC code with distance $\Omega(N/\log(N))$
  and dimension $\Omega(\log(N))$.
        \item Quantum LDPC code with distance $\Omega(N^{1-\alpha})$
  and dimension $\Theta(N^{\alpha})$ for every constant $\alpha > 0$.
   \end{enumerate}
\end{corollary}

\begin{proof}
  We always work in the constant degree regime so
  $\cC_0 \subseteq \mathbb{F}_2^d$ can be found by brute-force. When
  $\ell = \exp(\Theta(n))$, we use ~\cref{theo:main2} to construct $G$
  explicitly.  Thus, $N= n\ell$ and thus the circulant size and
  distance are both $\Theta(\ell) = \Theta(N/\log N)$.

  For $\ell \le 2^{n^{\delta_0}}$ with some fixed $\delta_0 \in
  (0,1)$, we can explicitly construct $G$ which is a
  $(\Z_\ell,\ell)$-lift by ~\cref{theo:main1} and by \cite{PK20},
  $\textup{T}(G,\mathcal{C}_0)$ has circulant size $\Theta(\ell)$ and
  $\log (N) \leq \log n + n^{\delta_0} \le 2n^{\delta_0}$ (for $n$
  sufficiently large) and thus, $\ell =O\left( N/(\log
  N)^{1/\delta_0} \right)$. Therefore, the construction works for any
  size less than $N/(\log N)^{1/\delta_0}$. This calculation also
  shows that we get quantum LDPC codes for any distance less than
  $N/(\log N)^{1/\delta_0}$.  So for a given $\alpha$, we take a base
  graph on $n = N^\alpha$ vertices and construct a $\ell =
  N^{1-\alpha} = n^{1/\alpha-1}$ lift. For any $\alpha$, this is a
  polynomial sized-lift and can be done via ~\cref{theo:main1}.
\end{proof}

\bibliographystyle{alphaurl}
\bibliography{macros,references}

\newcommand{\etalchar}[1]{$^{#1}$}
\begin{thebibliography}{ACKM19}

\bibitem[ACKM19]{ACKM19}
Naman Agarwal, Karthekeyan Chandrasekaran, Alexandra Kolla, and Vivek Madan.
\newblock On the {Expansion} of {Group}-{Based} {Lifts}.
\newblock {\em {SIAM} J. Discret. Math.}, 33(3):1338--1373, 2019.
\newblock \href {http://arxiv.org/abs/1311.3268} {\path{arXiv:1311.3268}},
  \href {https://doi.org/10.1137/17M1141047} {\path{doi:10.1137/17M1141047}}.

\bibitem[AHL02]{AHL02}
Noga Alon, Shlomo Hoory, and Nathan Linial.
\newblock The {Moore} bound for irregular graphs.
\newblock {\em Graphs and Combinatorics}, 18:53--57, 2002.
\newblock \href {https://doi.org/10.1007/s003730200002}
  {\path{doi:10.1007/s003730200002}}.

\bibitem[Alo21]{Alon21}
Noga Alon.
\newblock Explicit expanders of every degree and size.
\newblock {\em Combinatorica}, February 2021.
\newblock \href {http://arxiv.org/abs/2003.11673} {\path{arXiv:2003.11673}},
  \href {https://doi.org/10.1007/s00493-020-4429-x}
  {\path{doi:10.1007/s00493-020-4429-x}}.

\bibitem[AR94]{AR94}
Noga Alon and Yuval Roichman.
\newblock Random cayley graphs and expanders.
\newblock {\em Random Struct. Algorithms}, 5(2):271--285, 1994.
\newblock \href {https://doi.org/10.1002/rsa.3240050203}
  {\path{doi:10.1002/rsa.3240050203}}.

\bibitem[Bab74]{B74}
László Babai.
\newblock On the minimum order of graphs with given group.
\newblock {\em Canadian Mathematical Bulletin}, 17(4):467–470, 1974.
\newblock \href {https://doi.org/10.4153/CMB-1974-082-9}
  {\path{doi:10.4153/CMB-1974-082-9}}.

\bibitem[BATS08]{BT08}
Avraham Ben-Aroya and Amnon Ta-Shma.
\newblock A combinatorial construction of almost-ramanujan graphs using the
  zig-zag product.
\newblock In {\em Proceedings of the 40th ACM Symposium on Theory of
  Computing}, page 325–334, 2008.
\newblock \href {https://doi.org/10.1145/1374376.1374424}
  {\path{doi:10.1145/1374376.1374424}}.

\bibitem[BE21]{BreuckmannE21}
Nikolas~P. Breuckmann and Jens~N. Eberhardt.
\newblock Balanced {Product} {Quantum} {Codes}.
\newblock {\em IEEE Transactions on Information Theory}, 67(10):6653--6674,
  2021.
\newblock \href {http://arxiv.org/abs/2012.09271} {\path{arXiv:2012.09271}},
  \href {https://doi.org/10.1109/TIT.2021.3097347}
  {\path{doi:10.1109/TIT.2021.3097347}}.

\bibitem[BL06]{BL06}
Yonatan Bilu and Nathan Linial.
\newblock Lifts, discrepancy and nearly optimal spectral gap.
\newblock {\em Combinatorica}, 26(5):495--519, October 2006.
\newblock \href {https://doi.org/10.1007/s00493-006-0029-7}
  {\path{doi:10.1007/s00493-006-0029-7}}.

\bibitem[BM06]{BM06}
L.M.J. Bazzi and S.K. Mitter.
\newblock Some randomized code constructions from group actions.
\newblock {\em IEEE Transactions on Information Theory}, 52(7):3210--3219,
  2006.
\newblock \href {https://doi.org/10.1109/TIT.2006.876244}
  {\path{doi:10.1109/TIT.2006.876244}}.

\bibitem[BSS05]{BSS05}
L{\'{a}}szl{\'{o}} Babai, Amir Shpilka, and Daniel Stefankovic.
\newblock Locally testable cyclic codes.
\newblock {\em {IEEE} Trans. Inf. Theory}, 51(8):2849--2858, 2005.
\newblock \href {https://doi.org/10.1109/TIT.2005.851735}
  {\path{doi:10.1109/TIT.2005.851735}}.

\bibitem[Coh16]{Cohen16}
Michael~B. Cohen.
\newblock Ramanujan graphs in polynomial time.
\newblock In {\em Proceedings of the 57th IEEE Symposium on Foundations of
  Computer Science}, 2016.
\newblock \href {http://arxiv.org/abs/1604.03544} {\path{arXiv:1604.03544}},
  \href {https://doi.org/10.1109/FOCS.2016.37}
  {\path{doi:10.1109/FOCS.2016.37}}.

\bibitem[Con]{diag_url}
Keith Conrad.
\newblock {Simultaneous Ccommutativity Of Operators}.
\newblock
  \url{https://kconrad.math.uconn.edu/blurbs/linmultialg/simulcomm.pdf}.
\newblock [Online; accessed 7-September-2021].

\bibitem[CPW69]{CPW69}
C.L. Chen, W.W. Peterson, and E.J. Weldon.
\newblock Some results on quasi-cyclic codes.
\newblock {\em Information and Control}, 15(5):407--423, November 1969.
\newblock \href {https://doi.org/10.1016/s0019-9958(69)90497-5}
  {\path{doi:10.1016/s0019-9958(69)90497-5}}.

\bibitem[CS96]{CS96}
A.~R. Calderbank and Peter~W. Shor.
\newblock Good quantum error-correcting codes exist.
\newblock {\em Phys. Rev. A}, 54:1098--1105, Aug 1996.
\newblock \href {https://doi.org/10.1103/PhysRevA.54.1098}
  {\path{doi:10.1103/PhysRevA.54.1098}}.

\bibitem[EKZ20]{EKZ20}
Shai Evra, Tali Kaufman, and Gilles Z{\'{e}}mor.
\newblock Decodable quantum {LDPC} codes beyond the square root distance
  barrier using high dimensional expanders.
\newblock In {\em Proceedings of the 61st IEEE Symposium on Foundations of
  Computer Science}, pages 218--227. {IEEE}, 2020.
\newblock \href {http://arxiv.org/abs/2004.07935} {\path{arXiv:2004.07935}},
  \href {https://doi.org/10.1109/FOCS46700.2020.00029}
  {\path{doi:10.1109/FOCS46700.2020.00029}}.

\bibitem[Gal62]{Gal62}
R.~Gallager.
\newblock Low-density parity-check codes.
\newblock {\em IRE Transactions on Information Theory}, 8(1):21--28, 1962.
\newblock \href {https://doi.org/10.1109/TIT.1962.1057683}
  {\path{doi:10.1109/TIT.1962.1057683}}.

\bibitem[GW21]{GW21}
Oded Goldreich and Avi Wigderson.
\newblock Robustly self-ordered graphs: Constructions and applications to
  property testing.
\newblock In Valentine Kabanets, editor, {\em Proceedings of the 36th IEEE
  Conference on Computational Complexity}, volume 200 of {\em LIPIcs}, pages
  12:1--12:74. Schloss Dagstuhl - Leibniz-Zentrum f{\"{u}}r Informatik, 2021.
\newblock \href {https://doi.org/10.4230/LIPIcs.CCC.2021.12}
  {\path{doi:10.4230/LIPIcs.CCC.2021.12}}.

\bibitem[HHO21]{hho20}
Matthew~B. Hastings, Jeongwan Haah, and Ryan O'Donnell.
\newblock Fiber bundle codes: breaking the $n^{1/2} \mathrm{polylog} (n)$
  barrier for quantum {LDPC} codes.
\newblock In {\em Proceedings of the 52nd ACM Symposium on Theory of
  Computing}, pages 1276--1288. {ACM}, 2021.
\newblock \href {http://arxiv.org/abs/2009.03921} {\path{arXiv:2009.03921}},
  \href {https://doi.org/10.1145/3406325.3451005}
  {\path{doi:10.1145/3406325.3451005}}.

\bibitem[HLW06]{HooryLW06}
Shlomo Hoory, Nathan Linial, and Avi Wigderson.
\newblock Expander graphs and their applications.
\newblock {\em Bull. Amer. Math. Soc.}, 43(04):439–562, August 2006.
\newblock \href {https://doi.org/10.1090/S0273-0979-06-01126-8}
  {\path{doi:10.1090/S0273-0979-06-01126-8}}.

\bibitem[JM21]{JM21}
Akhil Jalan and Dana Moshkovitz.
\newblock Near-optimal cayley expanders for abelian groups.
\newblock {\em CoRR}, abs/2105.01149, 2021.
\newblock URL: \url{https://arxiv.org/abs/2105.01149}, \href
  {http://arxiv.org/abs/2105.01149} {\path{arXiv:2105.01149}}.

\bibitem[KW10]{KW10}
Tali Kaufman and Avi Wigderson.
\newblock Symmetric {LDPC} codes and local testing.
\newblock In Andrew~Chi{-}Chih Yao, editor, {\em Innovations in Computer
  Science - {ICS} 2010, Tsinghua University, Beijing, China, January 5-7, 2010.
  Proceedings}, pages 406--421. Tsinghua University Press, 2010.
\newblock URL:
  \url{http://conference.iiis.tsinghua.edu.cn/ICS2010/content/papers/32.html}.

\bibitem[LBM{\etalchar{+}}18]{LBMZX18}
Huaan Li, Baoming Bai, Xijin Mu, Ji~Zhang, and Hengzhou Xu.
\newblock Algebra-assisted construction of quasi-cyclic {LDPC} codes for {5G}
  new radio.
\newblock {\em IEEE Access}, 6:50229--50244, 2018.
\newblock \href {https://doi.org/10.1109/ACCESS.2018.2868963}
  {\path{doi:10.1109/ACCESS.2018.2868963}}.

\bibitem[LPS88]{LPS88}
Alexander Lubotzky, R.~Phillips, and Peter Sarnak.
\newblock Ramanujan graphs.
\newblock {\em Combinatorica}, 8:261--277, 1988.
\newblock \href {https://doi.org/10.1007/BF02126799}
  {\path{doi:10.1007/BF02126799}}.

\bibitem[Mar88]{Margulis88}
G.~A. Margulis.
\newblock Explicit group-theoretic constructions of combinatorial schemes and
  their applications in the construction of expanders and concentrators.
\newblock {\em Problemy Peredachi Informatsii}, 24(1):51--60, 1988.
\newblock URL: \url{http://mi.mathnet.ru/eng/ppi686}.

\bibitem[MOP20]{MOP20}
Sidhanth Mohanty, Ryan O'Donnell, and Pedro Paredes.
\newblock Explicit near-ramanujan graphs of every degree.
\newblock In {\em Proceedings of the 52nd ACM Symposium on Theory of
  Computing}, pages 510--523. {ACM}, 2020.
\newblock \href {http://arxiv.org/abs/1909.06988} {\path{arXiv:1909.06988}},
  \href {https://doi.org/10.1145/3357713.3384231}
  {\path{doi:10.1145/3357713.3384231}}.

\bibitem[Mor94]{Morgenstern94}
M.~Morgenstern.
\newblock Existence and explicit constructions of q + 1 regular {R}amanujan
  graphs for every prime power q.
\newblock {\em J. Comb. Theory Ser. B}, page 44–62, September 1994.
\newblock \href {https://doi.org/10.1006/jctb.1994.1054}
  {\path{doi:10.1006/jctb.1994.1054}}.

\bibitem[Nil91]{Nil91}
Alon Nilli.
\newblock On the second eigenvalue of a graph.
\newblock {\em Discrete Mathematics}, 91(2):207--210, 1991.
\newblock \href {https://doi.org/10.1016/0012-365X(91)90112-F}
  {\path{doi:10.1016/0012-365X(91)90112-F}}.

\bibitem[OW20]{OW20}
R.~O'Donnell and X.~Wu.
\newblock Explicit near-fully x-ramanujan graphs.
\newblock In {\em Proceedings of the 61st IEEE Symposium on Foundations of
  Computer Science}, pages 1045--1056, 2020.
\newblock \href {http://arxiv.org/abs/2009.02595} {\path{arXiv:2009.02595}},
  \href {https://doi.org/10.1109/FOCS46700.2020.00101}
  {\path{doi:10.1109/FOCS46700.2020.00101}}.

\bibitem[PK21a]{PK21}
Pavel Panteleev and Gleb Kalachev.
\newblock Asymptotically good quantum and locally testable classical ldpc
  codes, 2021.
\newblock \href {http://arxiv.org/abs/2111.03654} {\path{arXiv:2111.03654}}.

\bibitem[PK21b]{PK20}
Pavel Panteleev and Gleb Kalachev.
\newblock Quantum {LDPC} {Codes} with {Almost} {Linear} {Minimum} {Distance}.
\newblock {\em IEEE Transactions on Information Theory}, December 2021.
\newblock \href {http://arxiv.org/abs/2012.04068} {\path{arXiv:2012.04068}},
  \href {https://doi.org/10.1109/TIT.2021.3119384}
  {\path{doi:10.1109/TIT.2021.3119384}}.

\bibitem[Rao19]{Rao19}
Shravas Rao.
\newblock {A Hoeffding inequality for Markov chains}.
\newblock {\em Electronic Communications in Probability}, 24:1 -- 11, 2019.
\newblock \href {http://arxiv.org/abs/1806.11519} {\path{arXiv:1806.11519}},
  \href {https://doi.org/10.1214/19-ECP219} {\path{doi:10.1214/19-ECP219}}.

\bibitem[RU08]{RU08}
Tom Richardson and Ruediger Urbanke.
\newblock {\em Modern coding theory}.
\newblock Cambridge university press, 2008.
\newblock \href {https://doi.org/10.1017/CBO9780511791338}
  {\path{doi:10.1017/CBO9780511791338}}.

\bibitem[RVW00]{RVW00}
O.~Reingold, S.~Vadhan, and A.~Wigderson.
\newblock Entropy waves, the zig-zag graph product, and new constant-degree
  expanders and extractors.
\newblock In {\em Proceedings of the 41st IEEE Symposium on Foundations of
  Computer Science}, 2000.
\newblock \href {http://arxiv.org/abs/math/0406038}
  {\path{arXiv:math/0406038}}, \href {https://doi.org/10.1109/SFCS.2000.892006}
  {\path{doi:10.1109/SFCS.2000.892006}}.

\bibitem[Ste96]{Ste96}
Andrew Steane.
\newblock Multiple-particle interference and quantum error correction.
\newblock {\em Proceedings of the Royal Society of London. Series A:
  Mathematical, Physical and Engineering Sciences}, 452(1954):2551–2577, Nov
  1996.
\newblock \href {http://arxiv.org/abs/quant-ph/9601029v3}
  {\path{arXiv:quant-ph/9601029v3}}, \href
  {https://doi.org/10.1098/rspa.1996.0136} {\path{doi:10.1098/rspa.1996.0136}}.

\bibitem[Vad12]{Vadhan12}
Salil~P. Vadhan.
\newblock Pseudorandomness.
\newblock {\em Foundations and Trends® in Theoretical Computer Science},
  7(1–3):1--336, 2012.
\newblock \href {https://doi.org/10.1561/0400000010}
  {\path{doi:10.1561/0400000010}}.

\end{thebibliography}

\appendix

\section{Signed Non-backtracking Operator}\label{app:spectrum}

\subsection{Diagonalizing Non-backtracking Operator}

Let $\rho:Sym(l) \to \mathrm{GL}(\C^l)$ be the matrix representation of a
permutation. More concretely, given a permutation $\sigma \in \Sym(l)$
the map $\rho(\sigma)e_i = e_{\sigma\cdot i}$ where $\{e_1, \cdots,
e_l\}$ is the set of elementary basis vectors of $V = \C^l$. Since $H
\subseteq \Sym(l) \to  \mathrm{GL}(\C^l)$ it also gives a map on $H$ by
restriction. For example, let $P$ be the $l \times l$ permutation
matrix that maps $Pe_i = e_{i+1}$ where $i+1 $ is taken modulo
$l$. Then for $H= \Z_l$ and for $t \in \Z_l$, $\rho(t) = P^t$.

For a map $\rho$ as above and an extended signing $s$ , define a
generalized non-backtracking walk matrix in which for a non-zero entry
indexed by $(e_1,e_2)$ we replace $1$ by the block matrix
$\rho(s(e_2))$.

\begin{lemma}
  The non-backtracking walk matrix of the lifted graph is $B_{G(s)} = B_G(\rho)$. 
\end{lemma}

\begin{proof}
  In the lifted graph, the edges are of the form $[(u,i-s(u,v)), (v,i)] =: [u,v,i]$ and thus can be indexed by $E'\times [l]$.
  The non-backtracking walk matrix $B_{\hat{G}}$ would then have a non zero entry from $\left( [u,v,i] , [x,y,j] \right)$ iff $(v, i) = (x,j - s(x,y))$ and $(y, j ) \neq (u,i- s(u,v))$.
  Assume that the first condition is met i.e. $x = v$ and $j = i+s(x,y)$. If $y = u$, then $ i-s(u,v) = i  - s(y,x) = i+ s(x,y) = j $ and therefore, the second condition can't be met.
  This is just a longer way of saying that the lifts give a matching between $u \times [l]$ and $v \times [l]$.
  The implication of all this is that, $y$ has to be distinct from $u$ and thus the pair of edges $(u,v), (v,y)$ has a non-zero entry in $B_G$.
  Moreover, for every $i$ and every pair of edges $(u,v), (v,y)$ we have a non-zero entry for $(u,v, i) , (v,y, i+s(v,y))$ in $B_{G(s)}$
  and thus it can be written as a block matrix with the entry in $(u,v) , (v,y) $ equal to $\rho(s(v,y))$.
\end{proof}

Since the base graph $G$ and the signing $s$ will be fixed throughout,
we will drop the subscript to make reading less hurtful. We will need
the following well-known fact (see \cite[Thm. 5]{diag_url} for a nice
proof) that a collection of commuting matrices that are diagonalizable are also simultaneously diagonalizable. Since, $H$ is abelian, we have that $\{\rho(h)\}$ are commuting and since they are invertible, they clearly are diagonalizable. Thus, they simultaneously diagonalize, \ie there exists $F$ such that $\rho(h) = F\diag(\chi_1(h), \cdots,\chi_l(h)) F^{-1}$ where $\chi_i$ are characters of $H$.

\begin{corollary}\label{cor:decomposespec}
 If for $H$, the standard representation splits as $\rho = \oplus_i \chi_i$, then the non-backtracking walk matrix
  $B_{\hat{G}} = Q \diag(B(\chi_1), \cdots , B(\chi_t)) Q^{-1}$ and thus $\spec(B_{\hat{G}}) = \cup_i \spec(B(\chi_i))$.
  Moreover, if $H$ is transitive, then exactly one of the characters is trivial.
\end{corollary}
\begin{proof}
  To ease notation we write $B_G(\rho) = \sum M_{u,v}\otimes
  \rho(s(u,v))$ for some $M_{u,v}$. We have $\rho(s(u,v)) = F
  \diag(\chi_1(h), \cdots, \chi_l(h) ) F^{-1}$ and thus \[ B_G(\rho) =
  (I \otimes F) \sum M_{u,v}\otimes \diag(\chi_1(h), \cdots, \chi_l(h)
  ) (I \otimes F^{-1}) \] Let $|E| = N$ and let $T$ denote the
  permutation on $Nt$ that maps $T(i) := bt +(a+1) $ where $a ,b$ are
  the unique non-negative integers such that $0 \leq b <N$ $i-1 = aN
  +b$. It can then be seen that $\sum M_{u,v}\otimes \diag(\chi_1(h),
  \cdots, \chi_t(h) ) = T \diag\left( \sum M_{u,v}\otimes \chi_i(h)
  \right) T^{-1} $. Notice that $\sum M_{u,v}\otimes \chi_i(h) =
  B_G(\chi_i)$ and thus putting it together we have that for $Q = (I
  \otimes F)T$, $B_{G(s)} = Q \diag(B(\chi_1), \cdots , B(\chi_t))
  Q^{-1}$. The statement on the spectrum follows immediately.

  Since, the all-ones vector is clearly invariant under the standard
  representation, we have a copy of the trivial character $\chi_0$ in
  $\rho$. Let there be another vector $v \in \C^{\ell}$ that is
  invariant. Let $i \in \mathrm{supp}(v)$ and $j \not \in
  \mathrm{supp}(v)$. By transitivity, we have an $h$ such that $h\cdot
  i = j$ but then $h\cdot v \neq v$ which violates the invariance.
\end{proof}

\subsection{A Simple Consequence of Ihara-Bass}

We now slightly extend a claim in~\cite{MOP20} for general signings.

\begin{claim}\label{claim:o_donnel_obs}
  Let $A$ be the (signed) adjacency matrix of a $d$-regular graph.
  Suppose $f$ is an eigenvector of $A$ satisfying
  \begin{align*}
    Af = \left(\beta + \frac{d-1}{\beta}\right)f.
  \end{align*}
  Then $g(u,v) \coloneqq  (f(u) -\beta f(v))$
  (or in the signed case $g(u,v) \coloneqq A(u,v)^{-1} (f(u) -\beta \cdot A(u,v)f(v))$)
  is an eigenvector of the
  (signed) non-backtracking matrix $B$ with eigenvalue $\beta$.
\end{claim}

\begin{proof}
  Let $f$ and $g$ be as in the statement of the claim.
  Suppose for that $A$ and $B$ are not signed.
  Computing we have
  \begin{align*}
    (Bg)(u,v) &=\sum_{w\sim v, w\ne u} f(v) - \beta \cdot f(w)\\
              &= (d-1)f(v) - \sum_{w\sim v, w\ne u} \beta \cdot f(w)\\
              &= (d-1)f(v) + \beta \cdot f(u) - \beta \sum_{w\sim v} f(w)\\
              &= (d-1)f(v) + \beta \cdot f(u) - \beta (Af)(v)\\
              &= (d-1)f(v) + \beta \cdot f(u) - \beta \left(\beta + \frac{d-1}{\beta}\right)f(v)\\
              &= \beta (f(u)-\beta \cdot f(v)) = \beta \cdot g(u,v).
  \end{align*}
  Now suppose that $A$ and $B$ are signed.
  First note that $g$ is well-defined since for every entry $g(u,v)$ the
  pair $(u,v)$ is an orientation of an edge of the graph so it has
  a signing $A(u,v) \ne 0$.
  We have
  \begin{align*}
    (Bg)(u,v) &= \sum_{w\sim v, w\ne u} A(v,w) A(v,w)^{-1}(f(v) - \beta \cdot A(v,w) f(w))\\
              &= (d-1) f(v) - \beta \sum_{w\sim v, w\ne u} A(v,w)  f(w)\\
              &= (d-1) f(v) + \beta \cdot A(v,u) f(u) - \beta \sum_{w\sim v} A(v,w) f(w)\\
              &= (d-1) f(v)   + \beta \cdot A(v,u) f(u) - \beta (Af)(v)\\
              &= (d-1)f(v) + \beta \cdot  A(v,u) f(u) - \beta \left(\beta + \frac{d-1}{\beta}\right)f(v)\\
              &= \beta \cdot A(v,u) \left(f(u)-\beta \frac{1}{A(v,u)} f(v)\right),\\
              &= \beta \cdot A(u,v)^{-1} \left(f(u)-\beta \cdot A(u,v) f(v)\right) = \beta \cdot g(u,v),    
  \end{align*}
  where we used $A(v,u) = A(u,v)^{-1}$.
\end{proof}
  
\begin{corollary}\label{cor:ihara_bass_simple}
  Let $A$ be the (signed) adjacency matrix of a $d$-regular graph.
  Let $B$ be its (signed) non-backtracking operator.
  For any $\lambda > 2\sqrt{d-1}$, if $\;\rho_2(B) \le \lambda/2$, then $\rho_2(A) \le \lambda$. Hence, $\lambda(G) = \rho(A) \leq 2\max\set{\sqrt{d-1}, \rho_2(B)}$.
\end{corollary}

\begin{proof}
  We show via the contrapositive. Suppose that $f$ is eigenvector of $A$ with
  eigenvalue $\alpha$ such that $|\alpha| >  \lambda$.
  By possibly multiplying $A$ and $B$ by a phase (\ie $e^{i\theta}$), we can
  assume $\alpha$ is a non-negative real number.
  By~\cref{claim:o_donnel_obs}, we have that $\beta$ satisfying the equation
  $\beta^2-\alpha \beta + (d-1) =0$ is an eigenvalue of $B$. Considering
  the solution $\beta^+ = (\alpha + \sqrt{\alpha^2 -4(d-1)})/2$ and thus,
  we have $\beta^+ \ge \alpha/2 > \lambda/2 $.
\end{proof}

\vspace{-0.5cm}
\section{A Precise Implementation of DFS}\label{sec:dfs_implementation}

We now present the precise implementation of the depth-first search
algorithm (different versions could traverse an edge of an input graph
more than twice, but here we need exactly twice).

\begin{algorithm_simple}{DFS$(G,v$)}{connected graph $G$ and starting vertex $v$}{traversal $\mathcal{T}$ of $G$ and step types $\sigma$}\label{algo:dfs}
    \begin{enumerate}[topsep=0.3ex,itemsep=0.4ex,parsep=0.7ex,label={$\cdot$},leftmargin = 0.4cm]
      \item Color all vertices of $G$ with \textup{'green'}
      \item Traversal $\mathcal{T} = ()$
      \item Step types $\sigma = \textup{''}$
      \item Parent $\pi = \textup{null}$
      \item \textup{\textbf{DFSRec}}$(G,\mathcal{T},\pi,\sigma,v, e=\textup{null})$
      \item return $\mathcal{T},\sigma$
    \end{enumerate}
\end{algorithm_simple}

\begin{algorithm_simple}{DFSRec$(G,\mathcal{T},\pi,\sigma,v,e)$}{graph $G$, traversal $\mathcal{T}$, parent $\pi$, step types $\sigma$,
     vertex $v$ and traversed edge $e$}{Updated $\mathcal{T}$ (as side effect)}\label{algo:dfsrec}
  \begin{enumerate}[topsep=0.3ex,itemsep=0.4ex,parsep=0.7ex,label={$\cdot$},leftmargin = 0.4cm]
    \item $\mathcal{T}.\text{append}(e)$ if $e \ne \textup{null}$
    \item If $v$ is \textup{'green'}: \begin{enumerate}[itemsep=0.4ex,parsep=0.7ex,label={$\cdot$}]
    \item Color $v$ with \textup{'yellow'}
    \item For each neighbor $u$ of $v$ not colored \textup{'red'} and $u \ne \pi$: \begin{enumerate}[itemsep=0.4ex,parsep=0.7ex,label={$\cdot$}]
           \item $\sigma.\text{append}(\textup{'R'})$ \qquad\qquad (recursive step)
           \item \textup{\textbf{DFSRec}}$(G,\mathcal{T},\pi=v,\sigma, u, e=(v,u))$
           \end{enumerate}
      \item Color $v$ with  \textup{'red'}
    \end{enumerate}
    \item $\sigma.\text{append}(\textup{'B'})$ \qquad\qquad\qquad\qquad (backtrack step)
  \end{enumerate}
\end{algorithm_simple}

\noindent The key observation we need for this particular implementation of DFS
is the following.
\begin{remark}
  Let $G$ be a connected graph. The \textbf{DFS} algorithm traversals
  each edge of $G$ exactly twice: first in a recursive step and
  subsequently in a backtrack step.
\end{remark}

\end{document}